\documentclass[a4paper,12pt]{article}
\usepackage[utf8]{inputenc}
\usepackage{amsmath}
\usepackage{amsfonts}
\usepackage{amsthm}
\usepackage{amssymb}
\usepackage{lmodern}
\usepackage{enumitem}
\usepackage{tikz}
\usetikzlibrary{decorations.pathmorphing}
\usetikzlibrary{decorations.pathreplacing}
\usetikzlibrary{calc}
\usetikzlibrary{positioning}

\usepackage{fullpage}
\usepackage{xcolor}
\usepackage[english]{babel}
\usepackage[backgroundcolor=white, textwidth=2cm]{todonotes}

\theoremstyle{plain}
\newtheorem{lemma}{Lemma}
\newtheorem{theorem}[lemma]{Theorem}
\newtheorem{corollary}[lemma]{Corollary}

\title{
Walk refinement, walk logic, and the iteration number of the Weisfeiler-Leman algorithm
}

\author{Moritz Lichter\\
TU Kaiserslautern\\
\texttt{lichter@cs.uni-kl.de}
\and
Ilia Ponomarenko\\
St.\ Petersburg Department of Steklov Mathematical\\ Institute of the Russian Academy of Sciences\\
\texttt{inp@pdmi.ras.ru}
\and
Pascal Schweitzer\\
TU Kaiserslautern\\
\texttt{schweitzer@cs.uni-kl.de}
}

\usepackage{url}

\usepackage[colorlinks]{hyperref}
\definecolor{lightblue}{rgb}{0.5,0.5,1.0}
\definecolor{darkred}{rgb}{0.5,0,0}
\definecolor{darkgreen}{rgb}{0,0.5,0}
\definecolor{darkblue}{rgb}{0,0,0.5}

\hypersetup{colorlinks,linkcolor=darkred,filecolor=darkgreen,urlcolor=darkred,citecolor=darkblue}

\begin{document}

% (This is just for me = Moritz to display the PDF nicely on my dark theme)
%\InputIfFileExists{eyefriendly.tex} 

\newcommand{\countinglogic}[1]{\mathcal{C}_{#1}}
\newcommand{\walklogic}[1]{\mathcal{W}_{#1}}
\newcommand{\msetopen}[0]{\{\hspace{-3pt}\{}
\newcommand{\msetclose}[0]{\}\hspace{-3pt}\}}
\newcommand{\mset}[1]{\msetopen #1 \msetclose}
\newcommand{\msetcondition}[2]{\msetopen #1 \mid #2 \msetclose}
\newcommand{\set}[1] {\left\lbrace #1 \right\rbrace }
\newcommand{\setcondition}[2] {\left\lbrace #1 \mid #2 \right\rbrace }

\newcommand{\walkrefinement}[1]{#1_{\mathcal{W}}}
\newcommand{\walkrefinementb}[2]{#1_{\mathcal{W}[#2]}}
\newcommand{\wlrefinementname}{\mathsf{WL}}
\newcommand{\wlrefinement}[1]{#1_{\wlrefinementname}}

\newcommand{\walkcolor}[1]{\overline{#1}}

\newcommand{\refines}{\preceq}
\newcommand{\strictrefines}{\prec}
\newcommand{\isomorph}{\cong}

\newcommand{\sat}[0]{\vDash}
\newcommand{\nsat}[0]{\nvDash}
\newcommand{\bigand}{\bigwedge}
\renewcommand{\phi}{\varphi}
\renewcommand{\rho}{\varrho}

\newcommand{\partition}[1]{\pi(#1)}
\newcommand{\disjoiuntUnion}[0]{\dot{\cup}}
\newcommand{\nat}{\mathbb{N}}

\newcommand{\bigO}{\mathcal{O}}

\newcommand{\linearspan}[1]{\mathbb{C}#1}
\newcommand{\inducedalgebra}[1]{\langle #1 \rangle}
\newcommand{\allproducts}[1]{\widehat{#1}}
\newcommand{\fullmatrixalgebra}[1]{\mathsf{M}_{#1}(\mathbb{C})}
\newcommand{\algebra}[1]{\mathcal{#1}}
\newcommand{\productsOfLength}[2]{#1^{\leq #2}}
\newcommand{\rank}[1]{\mathrm{rank}(#1)}

\newcommand{\twisted}[1]{\tilde{#1}}
\newcommand{\straight}[1]{\bar{#1}}
\newcommand{\gadgetvertex}[1]{\straight{v}}
\newcommand{\repl}[1]{X(#1)}
\newcommand{\replt}[1]{\tilde{X}(#1)}
\newcommand{\repli}[1]{\hat{X}(#1)}

\newcommand{\grid}[2]{G_{#2}^{#1}}

\maketitle
%Forces page numbers
%\pagestyle{plain}
%\thispagestyle{plain}

\thanks{The research leading to these results has received funding from the European Research Council (ERC) under the European Union’s Horizon 2020 research and innovation programme (grant agreement No. 820148).}

 \begin{abstract}
We show that the $2$-dimensional Weisfeiler-Leman algorithm
stabilizes $n$-vertex graphs
after at most~$\bigO(n \log n)$ iterations.
This implies that
if such graphs are distinguishable
in~3-variable first order logic with counting,
then they can also be distinguished in this logic
by a formula of quantifier depth at most $\bigO(n \log n)$.

For this we exploit a new refinement based on counting walks
and argue that its iteration number differs
from the classic Weisfeiler-Leman refinement by at most a logarithmic factor.
We then prove matching linear upper and lower bounds
on the number of iterations of the walk refinement.
This is achieved with an algebraic approach
by exploiting properties of semisimple matrix algebras.
We also define a walk logic and a bijective walk pebble game
that precisely correspond to the new walk refinement.
 
\end{abstract}

\section{Introduction}
The classic Weisfeiler-Leman algorithm
is a tool that lies at the heart of algebraic combinatorics.
Developed first in 1968 to investigate symmetries of highly regular graphs,
it is routinely employed for the purpose of isomorphism testing.
Its development recently culminated in the WL2018 conference on
``Symmetry vs Regularity'' in algebraic graph theory in Pilsen~\cite{WL2018},
marking the 50 year anniversary of the algorithm.
Roughly speaking, the core idea of the classic ($2$-dimensional) algorithm
is to propagate structural information
regarding pairs~$(u,w)$ of vertices in a graph
by considering
for each possible third vertex $v$
the information already known
for the pairs $(u,v)$ and $(v,w)$.
To say equivalently,
the algorithm repeatedly classifies pairs~$(u,w)$ of vertices
according to the multiset of walks of length~$2$ from~$u$ to~$w$.
This process necessarily stabilizes after a finite number of iterations
and the corresponding stabilization is used to classify pairs of vertices. 

There is a close connection to a specific logic
namely the 3-variable fragment of first order logic with counting~\cite{CaiFI1992}.
Not only do the distinguishing power of the logic
and the distinguishing power of the algorithm agree,
there is also a close correspondence between the number of iterations
required by the algorithm and the quantifier depth required in the logic.

In this paper we are therefore interested in the number of iterations
after which the algorithm stabilizes.
Equivalently, we are interested in the maximum quantifier depth
needed in said first order logic to capture its expressibility on a graph of given size.
There is a trivial upper bound of~$n^2$ for this number,
since there are only~$n^2$ pairs of vertices
and thus a proper chain of partitions on the vertex pairs
(each partition finer than the previous one)
cannot be longer than~$n^2$.
With regard to lower bounds, Fürer~\cite{Furer2001}
proved that there are graphs
on which the stabilization number is in~$\Omega(n)$.
In~\cite{DBLP:conf/lics/KieferS16a},
using combinatorial techniques and a case distinction
into small and large vertex-color classes,
the currently best upper bound of~$\bigO(n^2/\log n)$
for the iteration number was proven.
In this paper we take an algebraic approach to the problem
and show the following upper bound.

\begin{theorem}
	\label{thm:wl-upper-bound}
The $2$-dimensional Weisfeiler-Leman algorithm
stabilizes after~$\bigO(n\log n)$ iterations
on graphs with $n$ vertices.
\end{theorem}

Via the above-mentioned correspondence between the Weisfeiler-Leman (WL) algorithm  and the logic~\cite{CaiFI1992} we obtain the following corollary.

\begin{corollary}
If two~$n$-vertex graphs
can be distinguished by a sentence in~$3$-variable first order logic with counting~$\countinglogic{3}$,
then there is also a $\countinglogic{3}$ sentence
of quantifier depth at most~$\bigO(n \log n)$ 
that distinguishes the two graphs.
\end{corollary}

To prove Theorem~1, we take an algebraic point of view and use a one to one correspondence between coherent configurations and coherent algebras~\cite{MR898557}. The WL algorithm produces the former as output, whereas the latter are semisimple matrix algebras
closed with respect to the Hadamard multiplication.

\subsection{Our technique}
Generalizing the idea of considering walks of length 2,
we consider a new type of refinement,
which we call the walk refinement.
In one iteration it distinguishes pairs of vertices
not only according to the multiset of 2-walks between them,
but rather considers the multiset of all walks of arbitrary length
between the two vertices.
Naturally, considering all walks cannot be weaker than considering 2-walks.
However, using arguments from linear algebra
it can be proven that it suffices to consider walks of bounded length.
This in turn can be used to argue that the walk refinement is subsumed
by a logarithmic number of  traditional 2-walk refinements.
In particular the two kinds of refinement yield the same stabilization
and their iteration numbers differ by at most a logarithmic factor. 

The cornerstone of our argument is then to show
that the number of iterations of the walk refinement
is at most linear in the number of vertices.
This is done by observing that
the result of walk refinement corresponds to a semisimple matrix algebra.
Multiple iterations of walk refinement
must therefore correspond to an increasing chain of semisimple subalgebras of a full matrix algebra. 
We can show that the length of such a chain is at most~$\bigO(n)$,
which gives a linear upper bound for the iteration number of repeated walk refinement.

Since our upper bound on the iteration number
of the WL refinement is tight up to a logarithmic factor,
the question arises whether the factor of~$\log n$ can be removed.
We show that walk refinement requires $\Theta(n)$ iterations
on the same graphs, 
for which Fürer~\cite{Furer2001} showed that the WL refinement requires $\Theta(n)$ iterations.
This leaves the problem open,
whether our method can be pushed further.

In our paper,
we also associate the walk refinement with a special Ehrenfeucht–Fraïssé type duplicator-spoiler game
and a variant of a counting logic.
We call them bijective walk pebble game and walk counting logic, respectively. 
They are suitable adaptations of a game and the logic $\countinglogic{3}$
that are associated with the classic Weisfeiler-Leman (2-walk) refinement.
The close correspondences between aspects of refinement algorithm,
game, and logic translate to our scenario (Theorem~\ref{thm:walk-refinement-logic-game-correspondence}).
In particular, we prove 
tight bounds on the length of shortest winning strategies
and optimal quantifier depth, respectively, of~$\Theta(n)$
(Theorem~\ref{thm:lower-bound-game}, Corollaries~\ref{cor:lower-bound-walk-logic} and~\ref{cor:lower-bound-walk-refinement}).

We should remark that while the combinatorial techniques
from~\cite{DBLP:conf/lics/KieferS16a}
showing the upper bound of~$\bigO(n^2/\log n)$
also translate to the setting without counting,
our techniques seem to strongly rely on counting,
since only the counting itself ensures
the correspondence to matrix algebras that we exploit.

\subsection{Related work}
Deep Stabilization (see~\cite{MR0543783}),
developed by Weisfeiler and Leman,
is a generalization of the classic $2$-dimensional WL algorithm.
It can in turn be seen as a restricted form
of the $k$-dimensional WL algorithm (for a suitable $k$)
in the sense of Babai (see~\cite{CaiFI1992}).
For each~$k\in \mathbb{N}$ both generalizations give a 
polynomial-time algorithm
that, as~$k$ increases, can distinguish more and more non-isomorphic graphs.
 
Over the course of the years
striking connections have been drawn between the Weisfeiler-Leman algorithm
and seemingly unrelated areas of research.
While at first it was unclear whether the algorithm
(for some~$k$) solves the graph isomorphism in polynomial time,
the seminal paper of Cai, Fürer, and Immerman~\cite{CaiFI1992}
answered this question in the negative.
Not only did they construct for each~$k\in \mathbb{N}$ graphs that
cannot be distinguished by the~$k$-dimensional version,
but they also exhibited a close connection to a logic with counting and
described the Ehrenfeucht–Fraïssé type duplicator-spoiler game mentioned above.
Optimal strategies in the game reflect precisely the outcome of the algorithm.
The precise logic in question here
is the~$k$-variable fragment of first order logic with counting.

Babai employs the~$k$-dimensional WL algorithm,
with~$k$ logarithmic in the input,
as a subroutine in his quasi-polynomial time algorithm for graph isomorphism testing~\cite{DBLP:conf/stoc/Babai16}. 

In the language of Grohe~\cite{MR3729479},
the Weisfeiler-Leman dimension of a graph~$G$ is the least number~$k$
for which the~$k$-dimensional WL algorithm
distinguishes~$G$ from every graph non-isomorphic to~$G$.
Equivalently, it is the number of variables
needed in said fixed-point logic with counting
to distinguish the graph from every non-isomorphic graph.
Grohe shows~\cite{MR3729479} that graphs with a forbidden minor
have bounded Weisfeiler-Leman dimension,
reassuring the polynomial-time solvability of the isomorphism problem
of such graphs~\cite{MR976178}. 
For recent developments that relate techniques
of the WL algorithm to group-CSP
(constrained satisfaction problems in which the constraints are cosets of a group)
we refer to~\cite{DBLP:conf/soda/BerkholzG17}.

Regarding bounds,
Berkholz and Nordstr\"{o}m~\cite{DBLP:conf/lics/BerkholzN16}
proved a lower bound on the number of iterations
of the~$k$-dimensional WL algorithm for finite structures.
Specifically, they show for sufficiently large~$k$ the existence of $n$-element relational structures
distinguished by the~$k$-dimensional WL algorithm
but for which~$n^{o(k/\log k)}$ iterations do not suffice.
For a different logic,
namely the $3$-variable existential negation-free fragment of first-order logic,
Berkholz also developed techniques to prove tight bounds~\cite{DBLP:conf/cp/Berkholz14}. In contrast to these bounds,
Fürer's lower bound~\cite{Furer2001} of~$\Omega(n)$ mentioned above is applicable to graphs and in fact also applies to all fixed dimensions~$k$.

\section{Preliminaries}

We denote with $[n]$ the set of numbers $\{1, \dots , n\}$ and with
$\mset{x_1, \dots , x_k}$ the multiset containing the elements $x_1$ to $x_k$.
In this paper we work with colored directed graphs ${G = (V, E, \chi)}$, where~$\chi\colon E \to C$ is a coloring function into some set~$C$ of colors. We will often consider complete directed graphs (with loops), i.e., the case~$E=V^2$.
We always require that $\chi$ assigns different colors
to loops than it does to other edges (i.e.,~${\chi(v,v) \neq \chi(u,w)}$ whenever~$u\neq w$).
For a tuple of $m > 1$ vertices $(v_1, \dots, v_m) \in V^m$
we set \[\walkcolor{\chi}(v_1, \dots , v_m) := (\chi(v_1, v_2), \chi(v_2, v_3), \dots, \chi(v_{m-1}, v_m))\]
and for a single vertex $v\in V$ we set $\chi(v) := \chi(v,v)$.

A coloring $\chi$ induces a partition $\partition{\chi}$ of the vertex pairs.
For two colorings~$\chi$ and~$\chi'$ we write
$\partition{\chi} \refines \partition{\chi'}$
to say that $\partition{\chi}$ is finer than $\partition{\chi'}$.
If not ambiguous we may only write $\chi \refines \chi'$.
If $\partition{\chi} = \partition{\chi'}$ we also write $\chi \equiv \chi'$.
We say that $\chi$ respects \emph{converse equivalence}
if ${\chi(u_1,v_1) = \chi(u_2, v_2)}$ implies ${\chi(v_1,u_1) = \chi(v_2, u_2)}$
for all $u_1,u_2,v_1,v_2 \in V$,
i.e., the color $\chi(u_1,v_1)$ determines $\chi(v_1, u_1)$.

A \emph{$k$-walk} or walk of \emph{length $k$} from $v_1$ to $v_{k+1}$ is a tuple $(v_1, \dots, v_{k+1}) \in V^{k+1}$.
Its color is $\walkcolor{\chi}(v_1, \dots, v_{k+1})$.
We say that tuples in $C^k$ are (potential) \emph{$k$-walk colors}
in $\chi$ and
omit the coloring if it is clear from the context.

A \emph{refinement} $r$ is a function
that for each graph $G$ and coloring $\chi$
yields a new coloring $\chi'$
such that $\chi' \refines \chi$.
Additionally, $r$ is required to be isomorphism invariant:
if we apply $r$  to two isomorphic, complete, and colored graphs
(i.e.,~there is an isomorphism respecting the partitions induced by the colorings)
then we obtain two new colorings that make the graphs isomorphic again.

We write $\chi_r$ for the application of the refinement $r$
to the coloring $\chi$ and $\chi_r^m$ for $m$ applications of $r$, i.e., ${\chi_r^{m+1}= {(\chi_r^m)}_r}$.
We denote with~$\chi_r^\infty$ the stable coloring,
i.e.,~the~$\chi_r^m$ for the smallest $m$
such that $\chi_r^m \equiv \chi_r^{m+1}$.
Let $G' = (V', E', \chi')$ be another colored complete graph,
and suppose $u, v \in V$ and $u',v' \in V'$.
We say that the refinement $r$
\emph{distinguishes} $(u,v)$ from $(u',v')$ in $m$ iterations
if $\chi_r^m(u,v) \neq {(\chi')}_r^m(u',v')$.
The refinement~$r$ \emph{distinguishes} $G$ and $G'$
in $m$ iterations, if the multiset of colors after $m$ iterations is different, that is
\[\msetcondition{\chi_r^m(u,v)}{u,v\in V} \neq \msetcondition{{(\chi')}_r^m(u',v')}{u',v'\in V'}.\]
We call the number of applications of $r$
needed to obtain the stable partition
the \emph{iteration number} of $r$.

Now let $G=(V,E)$ be an undirected and uncolored graph. 
We turn~$G$ to a colored complete graph
by defining a coloring $\chi\colon V^2\rightarrow \{-1,0,1\}$ setting~$\chi(v,v)=-1$ for all~${v \in V}$, as well as~$\chi(v,u)=1$ if~$v\neq u$ and~$(v,u)\in E(G)$, and setting~${\chi(v,u)=0}$ otherwise.
We refer to $\chi$ as the \emph{initial coloring} of $G$.
By construction the initial coloring respects converse equivalence.
A refinement distinguishes two undirected and uncolored graphs
if the respective initial colorings are distinguished by the refinement.
The analogous definition applies to vertex pairs.

\section{The Weisfeiler-Leman Refinement}

In this section we recall the $2$-dimensional WL refinement
and its connections to the counting logic with $3$ variables
and the bijective $3$-pebble game.
A~reader familiar with these notions should feel free to proceed to the next section.

Of particular interest in this paper is the
$2$-dimensional \emph{Weisfeiler-Leman} refinement~$\wlrefinementname$:
\[\wlrefinement{\chi}(u,v) := \msetcondition{\walkcolor{\chi}(u,w,v)}{w \in V}.\]
Intuitively, $\wlrefinementname$ refines the color of a vertex pair
with the colors of all triangles containing this pair.
This definition gives indeed a refinement:
because loops and non-loops always have distinct colors,
the presence of the color $\walkcolor{\chi}(u,u,v)$ (or $\walkcolor{\chi}(u,v,v)$) in the multiset ensures that
pairs colored differently remain colored differently after applying $\wlrefinementname$. In particular, we do not need to include the color $\chi(u,v)$ of the previous iteration in the new color explicitly
to ensure that $\wlrefinementname$ is a refinement.
In this paper, we refer with ``the Weisfeiler-Leman refinement'' always
to the $2$-dimensional WL refinement.

There is a close connection between the WL refinement,
the counting logic with three variables $\countinglogic{3}$,
and the so called bijective $3$-pebble game.

The logic $\countinglogic{3}$ provides counting existential quantification
and is limited to use three variables
(but they may be bound multiple times).
For~$k\in \mathbb{N}$ in general, $\countinglogic{k}$~formulas are defined
for a variable set $\mathcal{V}$ of size $k$ by the following grammar:
\[ \phi ::= 
x = y \mid
x \sim y \mid
\phi \land \phi \mid
\neg \phi \mid 
\exists^j x.~\phi  \hspace{0.5cm} x,y\in \mathcal{V}, j \in \nat.\]
The variables will be interpreted as vertices of an undirected graph,
$x=y$ expresses equality
and $x \sim y $ the edge relation of the graph.
A counting quantifier $\exists^j x.~\phi$ is satisfied if
there are at least $j$ distinct vertices
that satisfy~$\phi$.

The bijective $3$-pebble game is played by two players
called Spoiler and Duplicator on two undirected graphs ${G=(V,E)}$ and ${G'=(V',E')}$.
There are three pebble pairs $(p_1, q_1), (p_2, q_2), (p_3, q_3)$,
where the pebbles of the $i$-th pair are labeled with number $i$.
Initially, all pebbles are placed beside the graphs.
During the game the $p_i$ pebbles will be placed on the vertices of $G$
and the $q_i$ pebbles on the vertices of $G'$.
A~round of the game consists of the following moves:
\begin{enumerate}
	\item Spoiler picks up a pair of pebbles $(p_i, q_i)$.
	\item Duplicator chooses a bijection $f\colon  V \to V'$.
	\item Spoiler places $p_i$ on a vertex $v \in V$
	      and $q_i$ on $f(v) \in V'$.
\end{enumerate}
Spoiler wins the game after the~$m$-th round,
if mapping
the vertex covered by pebble~$p_i$
to the vertex covered by pebble~$q_i$
is not an isomorphism
of the subgraphs of~$G$ and~$G'$ induced by the vertices covered by pebbles.
Duplicator wins the game if Spoiler never wins the game.

The connection between the $2$-dimensional WL refinement,
the logic $\countinglogic{3}$, and the bijective $3$-pebble game is the following:
Let $(u,v) \in V^2$ and $(u',v') \in (V')^2$ be vertex pairs.
Then the following statements are equivalent~\cite{CaiFI1992, Hella96}:
\begin{itemize}
	\item If $u,u',v,v'$ are covered by $p_1,q_1,p_2,q_2$, respectively,
	      then Spoiler has a winning strategy finished after at most~$m$ rounds
	      (i.e.,~a way to win in at most $m$ rounds whatever Duplicator does).
	\item There is a $\countinglogic{3}$ formula $\phi(x,y)$
	      of quantifier depth~$m$ with
	      exactly two free variables,
	      such that $\phi$ holds on~$G$
	      when assigning~$u$ to~$x$ and~$v$ to~$y$,
	      but not on~$G'$ when assigning~$u'$ to~$x$ and~$v'$ to~$y$.
	\item After~$m$ iterations of the Weisfeiler-Leman refinement
	      (starting with the initial colorings of~$G$ and~$G'$)
	      the pairs $(u,v)$ and $(u',v')$ are colored differently.
\end{itemize}
It follows that the WL refinement distinguishes
exactly the same graphs as the logic $\countinglogic{3}$
(there is a sentence holding on one graph but not the other)
and as the bijective $3$-pebble game
(Spoiler has a winning strategy).

\section{Walk Refinement}

We now introduce a new refinement.
Suppose that ${G=(V,E,\chi)}$ is a complete and colored graph and recall that ${\walkcolor{\chi}(v_1, \dots , v_m) = (\chi(v_1, v_2), \chi(v_2, v_3), \dots, \chi(v_{m-1}, v_m))}$.
We define for $k \geq 2$ the \emph{$k$-walk refinement} to be
the function that for~$G$ and~$\chi$
gives the new coloring $\walkrefinementb{\chi}{k}$
defined by
\[\walkrefinementb{\chi}{k}(u,v) := \msetcondition{
	\overline{\chi}(u, w_1, \dots,  w_{k-1}, v)} {w_i \in V}.\]

Intuitively, the $k$-walk refinement refines the color of a vertex pair~$(u,v)$
with the color sequence of the traversed vertex pairs along walks, taken over all possible walks of length~$k$
from~$u$ to~$v$ (taken as multiset).
Note that,
since $\chi$ assigns different colors to loops,
the refinement implicitly also refines with respect to walks of shorter lengths $k' < k$
(indeed, the information is contained in the walks whose first $k-k'$ steps are of color $\chi(u)$, i.e., which are stationary at~$u$).
So the $k$-walk refinement is indeed a refinement,
i.e.~$\walkrefinementb{\chi}{k} \refines \chi$,
because walks of length $1$ are just the old colors.
It is easy to see that the $k$-walk refinement
is isomorphism invariant and preserves converse equivalence.

From what we just argued, obviously $\walkrefinementb{\chi}{k} \refines \walkrefinementb{\chi}{j}$ if $k \geq j$.
Also note that $2$-walk refinement is exactly the 2\nobreakdash-dimensional Weisfeiler-Leman refinement. 
Thus, for $k\geq 2$
\[\walkrefinementb{\chi}{k} \refines \wlrefinement{\chi} \refines \chi.\]

We argue next that $k$-walk refinement can be simulated with a logarithmic number of Weisfeiler-Leman refinements.

\begin{lemma}
	If $k \geq 2$ then \label{lem:simulate-walk-by-WL-refinement}
	$\wlrefinement{\chi}^{\lceil\log k\rceil} \refines \walkrefinementb{\chi}{k}$.
\end{lemma}

\begin{proof}
	Let~$C$ be the set of colors of $\chi \colon V^2 \to C$
	and let~$C_i$ be the set of colors of $\wlrefinement{\chi}^i \colon V^2 \to C_i$.
	
	We show by induction
	that after $i \geq 1$ iterations of the Weisfeiler-Leman refinement
	for each color $d \in C_i$
	there is a function $\bar{d}\colon C^{(2^i)} \to \nat$
	with the following property:
    For all $u,v \in V$ of color $\wlrefinement{\chi}^i(u,v) = d$
    and for all $2^i$-walk colors $(c_1, \dots c_{2^i}) \in C^{(2^i)}$ in $\chi$
	there are exactly 
	$\bar{d}(c_1, \dots, c_{2^i})$
	many $(c_1, \dots, c_{2^i})$ colored walks between $u$ and $v$ in $\chi$.
	In particular, the number of $(c_1, \dots, c_{2^i})$ colored walks
	is the same for all such $u$ and $v$.
	This implies $\wlrefinement{\chi}^{i} \refines \walkrefinementb{\chi}{2^i}$.
	
	For $i=1$, the Weisfeiler-Leman refinement assigns colors
	such that every color just contains the possible $2$-walk colors. So assume~$i>1$.
	
	Let $u,v\in V$ be vertices, $d_1, e_1, \dots, d_n, e_n \in C_{i}$
	be colors of~$\wlrefinement{\chi}^i$,
	$${d = \mset{(d_1,e_1),\dots, (d_{n}, e_{n})} \in C_{i+1}}$$
	be a color of~$\wlrefinement{\chi}^{i+1}$,
	and $(c_1, \dots, c_{2^{i+1}}) \in C^{(2^{i+1})}$ be a $2^{(i+1)}$-walk color.
	We set 
	\[\bar{d}(c_1, \dots, c_{2^{i+1}})
	  := \sum_{j \in [n]} \bar{d}_j (c_1, \dots, c_{2^i}) \cdot \bar{e}_j (c_{2^i+1}, \dots, c_{2^{i+1}}).\]
	By induction hypothesis $\bar{d}_i$ and $\bar{e}_i$
	yield the correct number of $2^i$-walks and
	so $\bar{d}$ yields the correct number of $2^{i+1}$ walks.
\end{proof}

This lemma corresponds to the known fact that in $\countinglogic{3}$
walks of length $k$ can be defined by a formula of quantifier depth $\lceil \log k \rceil$. These walks can be counted using counting quantifiers similarly.
Considering paths we see that the $k$-walk refinement cannot be simulated with less than a logarithmic number of Weisfeiler-Leman refinements, and in that sense the bound in the lemma is tight.
On the other hand the relation
can be strict, that is
${\wlrefinement{\chi}^{\lceil\log k\rceil} \strictrefines \walkrefinementb{\chi}{k}}$.
However, the Weisfeiler-Leman and $k$-walk refinement produce the same stable partition
because finitely many steps of one
subsume a single step of the other.

\begin{lemma}
	\label{lem:wl-and-walk-same-stable}
	If $k \geq 2$, then $\wlrefinement{\chi}^{\infty} \equiv \walkrefinementb{\chi}{k}^{\infty}$.
\end{lemma}
\begin{proof}
	Suppose that~$\walkrefinementb{\chi}{k}^{\infty} \equiv  \walkrefinementb{\chi}{k}^j$
	and $\wlrefinement{\chi}^{\infty} \equiv \wlrefinement{\chi}^j$
	for some suitable $j$.
	Then
$$\walkrefinementb{\chi}{k}^j \refines {\wlrefinement{\chi}^j \equiv \wlrefinement{\chi}^{j \cdot \lceil \log k \rceil} \refines \walkrefinementb{\chi}{k}^j}$$
	by Lemma~\ref{lem:simulate-walk-by-WL-refinement},
	and hence
	$\wlrefinement{\chi}^{\infty} \equiv \walkrefinementb{\chi}{k}^{\infty}$.
\end{proof}

We remark that it is possible that the partitions
produced by the Weisfeiler-Leman refinement
and the partitions produced by $k$-walk refinement
all disagree except for the stable partitions in the end
(for example this is the case for the graphs $\repl{\grid{2}{n}}$ for ${2\leq n \leq 10}$ defined in Section~\ref{sec:lower-bound} as shown by computer calculations with $k = n$).

We define the \emph{walk refinement} $\walkrefinement{\chi}$ as
the finest $k$-walk refinement. More precisely, we define it as $\walkrefinementb{\chi}{k}$ for the smallest~$k$,
for which
$\walkrefinementb{\chi}{k}$ induces the finest partition
over all choices of $k$.
We will prove in Section~\ref{sec:walk-refinment-iteration-number}
that $n^2$-walk refinement always produces this finest partition, thus~$k\leq n^2$. From that we will conclude that
${\walkrefinement{\chi} \equiv \walkrefinementb{\chi}{n^2}}$ and ${\wlrefinement{\chi}^{\bigO(\log n)} \refines \walkrefinement{\chi}}$. 
This will allow us
to bound the iteration number of the Weisfeiler-Leman refinement
by bounding the iteration number of walk refinement.

\section{Iteration Number of Walk Refinement}
\label{sec:walk-refinment-iteration-number}

In this section we show
that walk refinement stabilizes after $\bigO(n)$ iterations.
We interpret the partitions
produced by walk refinement as matrix algebras.
If walk refinement strictly refines the partition
then the algebra is strictly enlarged.
We obtain the linear bound by observing that
these algebras can be nested at most a linear number of times.

Throughout this section,
let $G=(V, E , \chi)$ be a complete and colored graph
with ${V = [n]}$
and
let $\chi \colon E \to C$ respect converse equivalence.

\subsection{Background on Matrix Algebras}

In this section we make use of standard material from representation theory,
see e.g.~\cite{zimmermann2014}.
Let $S$ be a set of $n\times n$ matrices over~$\mathbb{C}$.
We denote with $\linearspan{S}$ the $\mathbb{C}$-linear span of $S$
and with 
\[\productsOfLength{S}{k} := \setcondition{M_1 \cdot \ldots \cdot M_j }{j \leq k, M_i \in S \text{ for all } i \in [j]}\]
the set of all products of matrices in $S$ with at most $k$ factors.
Clearly ${\productsOfLength{S}{k} \subseteq \productsOfLength{S}{k+1}}$.
We write $\allproducts{S}$ for the union of all $\productsOfLength{S}{k}$.

For a color $c \in C$ we denote with~$M_c$
the $n\times n$ \emph{color~$c$ adjacency matrix},
that is~$(M_c)_{ij}=1$ if~$\chi(i,j)=c$ and~$(M_c)_{ij}=0$ otherwise.
The set of all color adjacency matrices is denoted by $M_\chi$.
The coloring $\chi$  thereby induces an $n \times n$ matrix algebra over the complex numbers:
\[ \inducedalgebra{\chi} := \linearspan{\allproducts{M}_\chi}.\]
The algebra $\inducedalgebra{\chi}$ is closed under (conjugate) transposition
because $\chi$ respects converse equivalence.

We write $\fullmatrixalgebra{k}$ for the (full) matrix algebra
of all $k \times k$ matrices over the complex numbers.
It is a well-known fact that a matrix algebra~${\algebra{A} \subseteq \fullmatrixalgebra{n}}$
closed under conjugate transposition is always semisimple.
Indeed, if $M$ is in the Jacobson radical of $\algebra{A}$,
so is $M^*M$.
But $M^*M$ is diagonalizable (because it is Hermitian)
and nilpotent (because the radical is nilpotent, Lemma 1.6.6 in~\cite{zimmermann2014})
and hence $M^*M = 0$ and so $M = 0$.
Then the radical itself is $0$,
which is one characterization of semisimplicity.

By the theorem of Wedderburn (Corollary~1.4.17 in~\cite{zimmermann2014})
a semisimple matrix algebra $\algebra{A} \subseteq \fullmatrixalgebra{n}$
is always isomorphic to a direct sum of full matrix algebras, that is
\[\algebra{A} \isomorph \fullmatrixalgebra{a_1} \oplus \dots \oplus \fullmatrixalgebra{a_k},\]
for some positive integers~$k$ and~$a_1,\ldots, a_k$.
The direct sum decomposition is unique up to reordering.
We will prove a bound on the length of proper chains
${\algebra{A}_1 \subsetneq \dots \subsetneq \algebra{A}_m \subseteq \fullmatrixalgebra{n}}$
of semisimple matrix algebras.
This is the essential theorem to bound the iteration number of walk refinement:

\begin{theorem}
	\label{thm:matrix-alegbra-chain-length}
	Let 
	$\algebra{A}_1 \subsetneq \dots \subsetneq \algebra{A}_m \subseteq \fullmatrixalgebra{n}$
	be a chain of semisimple strict subalgebras.
	Then $m \leq 2n$.
\end{theorem}
To prove the theorem we need several auxiliary lemmas, which may be  	self-evident for a reader familiar with the theory of semisimple algebras.
They show that such chains behave well
with respect to the direct sum decompositions of the $\algebra{A}_i$.

\begin{lemma}
	\label{lem:dimensions-direct-sum}
	If there is an algebra monomorphism
	\[\fullmatrixalgebra{a_1} \oplus \dots \oplus \fullmatrixalgebra{a_k} 
	\to \fullmatrixalgebra{m},\]
	then $\sum_{i=1}^k a_i \leq m$.
\end{lemma}
\begin{proof}
	For each~$i$, there are exactly~$a_i$ diagonal matrice
	in~$\fullmatrixalgebra{a_i}$ with one entry $1$ and all other $0$.
	These matrices are nonzero, idempotent, and pairwise orthogonal
	(i.e., the product of any two of them is 0).
	It follows that the direct sum, and hence the monomorphism image,
	contains a set of $d=\sum_{i=1}^k a_i$ 
	nonzero, idempotent, and pairwise orthogonal elements.
	
	Let~$M_1, \dots, M_d \in \fullmatrixalgebra{m}$
	be the matrices of this set.	
	Since the $M_i$ are nonzero, each has rank at least 1.
	Suppose ${i \neq j \in [d]}$.
	Then $M_i+M_j$ is idempotent,
	orthogonal to all other $M_{\ell}$,
	and has  $\rank{M_i+M_j} = \rank{M_i} + \rank{M_j}$
	(see e.g.~Theorem~IV.12 in~\cite{Adrian1937}).
	Because the maximal rank of an~$m\times m$ matrix is $m$,
	it follows by induction that
	$d=\sum_{i=1}^k a_i \leq \rank{\sum_{i=1}^k M_i} \leq m$.
\end{proof}

\begin{lemma}
	\label{lem:homomorphisms-direct-sum}
	Suppose $a_1, \dots, a_k, {b}_1, \dots, {b}_{\ell} \in \nat$ and let 
	\[\phi \colon \bigoplus_{i=1}^k \fullmatrixalgebra{a_i} \to 
	         \bigoplus_{j=1}^\ell \fullmatrixalgebra{b_j}\]
	be an algebra monomorphism.
	Then for every $i \in [k]$ there is a $j\in [\ell]$
	such that~$\pi_j \circ \phi$ maps~$\fullmatrixalgebra{a_i}$
	injectively into~$\fullmatrixalgebra{b_j}$,
	where $\pi_j$ is the projection onto the $j$-th component $\fullmatrixalgebra{b_j}$.
\end{lemma}
\begin{proof}
	Full matrix algebras are simple,
	i.e.,~contain no proper nontrivial two-sided ideals. 
	For a simple algebra $\algebra{A}$,
	any homomorphism $\psi \colon \algebra{A} \to \algebra{B}$
	into some algebra $\algebra{B}$
	is either injective or zero
	(otherwise, $\ker(\psi)$ is a proper nontrivial two-sided ideal of $\algebra{A}$).
	Suppose $i \in [k]$.
	The map $\pi_j \circ \phi_i \colon \fullmatrixalgebra{a_i} \to \fullmatrixalgebra{b_j}$ is an algebra monomorphism for all $j \in [\ell]$,
	where $\phi_i$ is the restriction of $\phi$ to
	the $i$-th component $\fullmatrixalgebra{a_i}$.
	Now, $\pi_j \circ \phi_i$ must be injective for some~$j$,
	because if all $\pi_j \circ \phi_i$ were zero,
	$\phi$ was not injective.
\end{proof}

\begin{lemma}
	\label{lem:agreeing-sets-step}
	Let $\algebra{A} \subseteq \algebra{B} \subseteq \fullmatrixalgebra{n}$
	be two semisimple matrix algebras with direct sum decompositions
	\begin{align*}
	\algebra{A} &\isomorph \fullmatrixalgebra{a_1} \oplus \dots \oplus \fullmatrixalgebra{a_k} \text{ and }\\
	\algebra{B} &\isomorph \fullmatrixalgebra{b_1} \oplus \dots \oplus
	\fullmatrixalgebra{b_\ell}.
	\end{align*}
	Then~$ 2(\sum_{i = 1}^{k} a_{i})-k \leq 2(\sum_{j = 1}^{\ell} b_{j})-\ell$ with equality exactly if~$\algebra{A} = \algebra{B}$.
\end{lemma}
\begin{proof}
	First, we pick an algebra monomorphism 
	${\phi \colon \bigoplus_{i=1}^k \fullmatrixalgebra{a_i} 
		\to \bigoplus_{j=1}^\ell \fullmatrixalgebra{b_j}}$.
	Second, for each ${i\in [k]}$ we choose an~${f(i)=j}$
	such that~${\pi_j \circ \phi}$ maps $\fullmatrixalgebra{a_i}$
	injectively into~$\fullmatrixalgebra{b_j}$.
	Such choices exist by  Lemma~\ref{lem:homomorphisms-direct-sum}.
	For each~$j\in [{\ell}]$ we obtain a monomorphism%
	~$\bigoplus_{i\in f^{-1}(j)} \fullmatrixalgebra{a_i} \to \fullmatrixalgebra{b_j}$
	by restricting $\pi_j \circ \varphi$.
	From Lemma~\ref{lem:dimensions-direct-sum}
	it now follows that~$b_j\geq \sum_{i\in f^{-1}(j)} a_i$.
	Then, to show that
	${2(\sum_{i = 1}^{k} a_{i})-k \leq 2(\sum_{j = 1}^{\ell} b_{j})-\ell}$,
	simply observe that for each $j\in [\ell]$ we have
	${2 b_j-1 \geq 2 (\sum_{i\in f^{-1}(j)} a_i)-|f^{-1}(j)|}$
	(since $b_j>0$)
	and that summing up over all~$j$ yields the desired equation.
	Finally, consider the case of equality and let $j \in [\ell]$.
	Then $2b_j - 1 = 2(\sum_{i \in f^{-1}(j)} a_i ) - |f^{-1}(j)|$
	and because $b_j \geq \sum_{i \in f^{-1}(j)} a_i$
	it follows that $|f^{-1}(j)| = 1$
	and $b_j = a_{f^{-1}(j)}$.
	Thus~$f$ is a bijection
	satisfying $a_i = b_{f(i)}$ for all $i \in [k]$
	and~$\varphi$ is an isomorphism.
\end{proof}

We now conclude the proof of Theorem~\ref{thm:matrix-alegbra-chain-length}.

\begin{proof}[Proof of Theorem~\ref{thm:matrix-alegbra-chain-length}]
	For all ${i \in [m]}$ suppose that
	${\algebra{A}_i \isomorph \bigoplus_{j = 1}^{n_i} \fullmatrixalgebra{a_{ij}}}$ and
	define \linebreak[4]$s_i:= 2(\sum_{j = 1}^{n_i} a_{ij})-n_i$.
	Then~$s_m\leq 2n-1$ by Lemma~\ref{lem:dimensions-direct-sum}. 
	Moreover, by Lemma~\ref{lem:agreeing-sets-step} for all~$i\in [m-1]$ we have $s_i\leq s_{i+1}$ and in fact even $s_i< s_{i+1}$ since equality would imply~$\algebra{A}_i=\algebra{A}_{i+1}$.
	Thus~$m\leq 2n$.
\end{proof}

\subsection{Matrix Algebras and Walk Refinement}
We say that a matrix $M \in \fullmatrixalgebra{n}$ distinguishes
$(u_1, v_1)$ from $(u_2, v_2)$
if $M_{u_1,v_1} \neq M_{u_2, v_2}$
and that a set $S \subseteq \fullmatrixalgebra{n}$ distinguishes
$(u_1, v_1)$ from $(u_2, v_2)$
if $S$ contains a matrix distinguishing them.

We now show that with one iteration of walk refinement
we can distinguish the same vertex pairs
as with the induced algebra $\inducedalgebra{\chi}$.

Let $c_1, c_2 \in C$ be colors. Then $(M_{c_1} \cdot M_{c_2})_{u,v}$
is the number of $(c_1, c_2)$ colored walks from $u$ to $v$.
In general, let $c_1, \dots, c_k$ be colors.
Then $(M_{c_1} \cdot \ldots \cdot M_{c_k})_{u,v}$
is the number of $(c_1, \dots, c_k)$ colored walks from $u$ to $v$.
Because walk refinement and the induced algebra
essentially count colored walks,
they distinguish the same vertex pairs:

\begin{lemma}
	\label{lem:walk-refinment-distinguishes-iff-induces-algebra}
	Let $u_1,v_1,u_2,v_2 \in V$.
	The walk refinement $\walkrefinement{\chi}$ distinguishes
	$(u_1,v_1)$ from $(u_2,v_2)$
	if and only if
	the induced algebra $\inducedalgebra{\chi}$
	distinguishes them.
\end{lemma}
\begin{proof}
	On the one hand suppose that walk refinement distinguishes the vertices, \linebreak[4] i.e.~${\walkrefinement{\chi}(u_1, v_1) \neq \walkrefinement{\chi}(u_2, v_2)}$.
	Then there is a sequence of colors $c_1, \dots, c_k$
	such that the number of $(c_1, \dots , c_k)$ colored walks
	between~$u_1$ and~$v_1$ is different from the number
	of such walks between~$u_2$ and~$v_2$.
	Hence $M_{c_1} \cdot \ldots \cdot M_{c_k}$ distinguishes the two vertex pairs.
	
	On the other hand let $M \in \inducedalgebra{\chi}$ distinguish $(u_1, v_1)$ and $(u_2, v_2)$.
	The matrix $M$ is a linear combination
	of products of color adjacency matrices:
	\[M =\sum_{i = 1}^{m} z_i \prod_{j = 1}^{k_i} M_{c_{i j}}\]
	where~$z_i\in \mathbb{C}$  and $c_{i j} \in C$ for all $i \in [m]$ and $j \in [k_i]$.
	There must be an $i$
	such that $\prod_{j = 1}^{k_i} M_{c_{i j}}$
	distinguishes $(u_1, v_1)$ and $(u_2, v_2)$,
	because $M$ distinguishes them.
	Hence the number of $(c_{i 1}, \dots,  c_{i k_i})$
	colored walks between $u_1$ and $v_1$
	is different from the number of such walks
	between $u_2$ and $v_2$ and 
	the pairs are distinguished by walk refinement.
\end{proof}

\begin{corollary}
	\label{cor:algebra-larger-if-refinement-finer}
	Either $\inducedalgebra{\chi} \subsetneq \inducedalgebra{\walkrefinement{\chi}}$
	or $\chi \equiv \walkrefinement{\chi}$. 
\end{corollary}

The induced algebra gets strictly larger
if the partition induced by walk refinement gets strictly finer.
We obtain the bound on the walk refinement iterations,
because the algebras can be nested only $2n$ many times.

\begin{theorem}\label{thm:walk:ref:stab}
	Walk refinement stabilizes in $2n$ iterations.
\end{theorem}
\begin{proof}
	Assume that $m$ is the smallest number
	such that ${\walkrefinement{\chi}^m \equiv \walkrefinement{\chi}^\infty}$.
	Because walk refinement preserves converse equivalence,
	the algebras 
	$
	\inducedalgebra{\chi},
	\inducedalgebra{\walkrefinement{\chi}}, 
	\dots,
	\inducedalgebra{\walkrefinement{\chi}^m}
	$ are semisimple.
	Then from Corollary~\ref{cor:algebra-larger-if-refinement-finer}
	it follows that
	$
	 {\inducedalgebra{\chi} \subsetneq
	 \inducedalgebra{\walkrefinement{\chi}} \subsetneq 
	 \dots \subsetneq
	 \inducedalgebra{\walkrefinement{\chi}^m}}$
	and from Theorem~\ref{thm:matrix-alegbra-chain-length}
	that $m \leq 2n$.
\end{proof}

To obtain a bound on the iteration number of the Weisfeiler-Leman refinement,
it remains  to relate the Weisfeiler-Leman refinement and the walk refinement.

\begin{lemma}
\label{lem:walk-refinment-walk-length-bound}
 ${\walkrefinement{\chi} \equiv \walkrefinementb{\chi}{n^2}}$ and $\wlrefinement{\chi}^{\bigO(\log n)} \refines \walkrefinement{\chi}$.
\end{lemma}

\begin{proof}
	We first show ${\walkrefinement{\chi} \equiv \walkrefinementb{\chi}{n^2}}$.
	A close inspection of the proof of Lemma~\ref{lem:walk-refinment-distinguishes-iff-induces-algebra}
	shows that $\linearspan{\productsOfLength{M_\chi}{k}}$ distinguishes
	the same vertex pairs as $k$-walk refinement.
	It suffices to show that $\linearspan{\allproducts{M}_\chi} = \linearspan{\productsOfLength{M_\chi}{n^2}}$, which implies
	$\inducedalgebra{\chi} = \linearspan{\productsOfLength{M_\chi}{n^2}}$
	and $\walkrefinement{\chi} \equiv \walkrefinementb{\chi}{n^2}$.
	
	The argument is well-known:
	Let $S$ be a set of $n\times n$ matrices, then clearly\linebreak[4]
	${\dim \linearspan{\productsOfLength{S}{k}} 
	 \leq \dim \linearspan{\productsOfLength{S}{k+1}}}$.
	If ${\dim \linearspan{\productsOfLength{S}{k}} 
	   = \dim \linearspan{\productsOfLength{S}{k+1}}}$,
	then ${\linearspan{\productsOfLength{S}{k}}
	   = \linearspan{\productsOfLength{S}{j}}}$ for all ${j \geq k}$.
    Hence ${\linearspan{\allproducts{S}} 
       = \linearspan{\productsOfLength{S}{n^2}}}$
    because the dimension can be at most $n^2$.
    
    Now $\wlrefinement{\chi}^{\bigO(\log n)} \refines \walkrefinement{\chi}$
    follows by Lemma~\ref{lem:simulate-walk-by-WL-refinement}.
\end{proof}

\begin{proof}[Proof of Theorem~\ref{thm:wl-upper-bound}]
	Finally, combining Theorem~\ref{thm:walk:ref:stab}
	with Lemmas~\ref{lem:wl-and-walk-same-stable} 
	and~\ref{lem:walk-refinment-walk-length-bound}
	proves Theorem~\ref{thm:wl-upper-bound},
	namely that the Weisfeiler-Leman refinement
	stabilizes in $\bigO(n \log n)$ iterations.
\end{proof}

We argued that the length of the involved matrix algebras
(the smallest number $k$, such that
$\linearspan{\productsOfLength{S}{k}} = \linearspan{\allproducts{S}}$) is at most~$n^2$. 
We remark that there is even an $\bigO(n \log n)$ bound~\cite{Shitov2018}
for the length of matrix algebras.
But this bound does not improve
our bound on Weisfeiler-Leman iterations asymptotically.

\section{Walk Counting Logic}

The Weisfeiler-Leman refinement can distinguish the same graphs as
the counting logic~$\countinglogic{3}$.
More strongly the number of Weisfeiler-Leman iterations needed
to distinguish two vertex pairs equals the minimum
quantifier depth of a formula to distinguish them.
As we have already seen, for $k \geq 2$
the $k$\nobreakdash-walk refinement distinguishes the same vertex pairs, too.
But the required iterations of walk refinement do not correspond to the quantifier depth of $\countinglogic{3}$.
We now introduce a logic $\walklogic{k}$ we call \emph{$k$-walk counting logic} for which such a correspondence holds.
The logic is defined for undirected and uncolored graphs.
We could relax the restriction to directed and colored graphs
respecting converse equivalence as in the previous section
but this is not needed in this paper.

The logic~$\walklogic{k}$ uses a set~$\mathcal{V}$ of $k+1$ variables
and every~$\walklogic{k}$ formula~$\phi$ has at most two free
variables, which we indicate using the notation~$\phi(x,y)$.
The~$\walklogic{k}$ formulas with free variables $z_1,z_{k+1} \in \mathcal{V}$
are defined according to the grammar
\begin{align*}
	\phi(z_1,z_{k+1}) ::=~& 
		z_1 = z_{k+1} \mid
		z_1 \sim z_{k+1} \mid\\
		&\phi(z_1,z_{k+1}) \land \phi(z_1,z_{k+1}) \mid
		\neg \phi(z_1,z_{k+1}) \mid \\
		&\exists^j (z_2, \dots, z_{k}).~
		    \bigand_{i \in [k]} \phi(z_i, z_{i+1})
\end{align*}
where $z_i\in \mathcal{V}$ and $j \in \nat$.
The variables $z_1, \dots,  z_{k+1} \in \mathcal{V}$
in the existential quantifier are required to be pairwise distinct.
We call the existential quantifier above a \emph{$k$-walk quantifier}.
As usual with grammars,
the subformulas of a $k$\nobreakdash-walk quantifier~$\phi(z_i, z_{i+1})$, which are non-terminals, 
can be replaced by different formulas.
The $k$-walk quantifier above is satisfied,
if there are at least~$j$ distinct tuples $(v_2, \dots, v_{k})$ of vertices
satisfying the rest of the formula.

Note that sentences can for example be obtained by setting
$\phi(z_1,z_2)$ to be the formula $z_2 = z_2$ and setting $\phi(z_{k}, z_{k+1})$ to be $z_{k} = z_{k}$.
This restricts the top most walk quantifier to
quantify over walks of length $k-2$.
The restriction could be relaxed, but that would complicate the definition
and is not needed for our purpose.

Syntactically, $\walklogic{k}$ is not a subset of $\walklogic{k+1}$,
but obviously for every $\walklogic{k}$~formula 
there is an equivalent  $\walklogic{k+1}$~formula.

Let~$\phi(x,y)$ be a~$\walklogic{k}$ formula,
$G =(V,E)$ an undirected graph, and ${u, v \in V}$ be vertices.
By~$\phi_G(u,v)$ we denote the truth value of~$\phi$ on~$G$
when assigning $u$~to $x$~and $v$~to $y$.
We omit the subscript if the graph is clear from the context.

Let ${G_1 = (V_1, E_1)}$ and ${G_2 = (V_2, E_2)}$ be two undirected graphs,
$u_1, v_1 \in V_1$, and $u_2, v_2 \in V_2$.
We say that 
\begin{itemize}
	\item a $\walklogic{k}$ formula $\phi(x,y)$ \emph{distinguishes} $(u_1, v_1)$ from $(u_2, v_2)$,
	if $\phi_{G_1}(u_1, v_1)$ is different from $\phi_{G_2}(u_2, v_2)$,
	\item a $\walklogic{k}$ sentence $\phi$ \emph{distinguishes} $G_1$ from $G_2$ if $\phi$ has different truth values on~$G_1$ and~$G_2$, and
	\item $\walklogic{k}$ distinguishes~$G_1$ from~$G_2$ 
	if there is a  $\walklogic{k}$ sentence distinguishing them.
\end{itemize}
For a coloring $\chi \colon V^2 \to C$,
we also say that a $\walklogic{k}$ formula $\phi(x,y)$
\emph{identifies} a color $c \in C$ in $\chi$,
if $\phi(u,v)$ holds if and only if $\chi(u,v) = c$ for all $u,v \in V$.

We call the union of the $\walklogic{k}$ logics for all $k \in \nat$
\emph{walk counting logic}.
Consequently, the number of variables in the walk counting logic is unbounded.
The definitions for distinguishing vertex pairs and graphs
for walk counting logic are analogous to $\walklogic{k}$.

We now show that with $\walklogic{k}$ formulas
of quantifier depth $m$
one can distinguish at least as many vertex pairs 
as with $m$ iterations of $k$-walk refinement.

\begin{lemma}
	\label{lem:formula-identifing-walk-refinement-color}
	Let $G = (V,E)$ be an undirected graph,
	$\chi$ the initial coloring for $G$,
	and $c$ a color produced by
	$m$ iterations of $k$-walk refinement.
	Then there is a $\walklogic{k}$
	formula $\phi(x,y)$ of quantifier depth $m$
	identifying $c$ in $\walkrefinementb{\chi}{k}^m$.
	Moreover, $\phi(x,y)$ only depends on $n$ and $c$
	(but not on $G$).
\end{lemma}
\begin{proof}
	If $m=0$, then $c$ stands either for loop, edge, or non edge,
	which is identified by the formulas $x=y$, $x \sim y$, and $x \not\sim y$.
	
	Let $c$ be a color in the $(m+1)$-th iteration.
	Hence $c$ is a multiset of $k$-walk colors
	in $\walkrefinementb{\chi}{k}^m$.
	Let $(c_1, \dots , c_k)$ occur with multiplicity $j$ in $c$.
	Then there are formulas $\phi_{c_i}$ 
	of quantifier depth $m$
	identifying $c_i$ in $\walkrefinementb{\chi}{k}^m$
	by induction hypothesis.
	The formula 
	\[\phi_{(c_1, \dots , c_k)} (z_1, z_{k+1}) := \exists^{j} z_2, \dots , z_{k}.~
	\bigand_{i \in [k]} \phi_{c_{i}}(z_i, z_{i+1})\]
	holds for vertices $u$ and $v$ assigned to $z_1$ and $z_{k+1}$ respectively
	if and only if
	there are at least $j$~many $(c_1, \dots , c_k)$ colored walks 
	from $u$ to $v$ in $\walkrefinementb{\chi}{k}^m$.
	
	Then the conjunction over all $(c_1, \dots , c_k)$ in $c$
	identifies $c$ in $\walkrefinementb{\chi}{k}^{m+1}$
	and is of quantifier depth $m+1$.
\end{proof}

There is a technical detail that, when one is interested in
distinguishing graphs rather than distinguishing vertex pairs, one (sometimes) needs an additional quantifier.
This is the case because
a refinement distinguishes two graphs after $m$ applications,
if the multisets of colors of both graphs are different.
Hence, there is a hidden quantifier
(saying that there is a color,
that occurs with different multiplicity in both graphs).

\begin{lemma}
	\label{lem:refinement-implies-logic-step}
	If $m$ iterations of $k$-walk refinement
	distinguish two graphs $G_1$ and $G_2$,
	then a $\walklogic{k}$ sentence
	of quantifier depth $m+1$ (respectively $m+2$ if $k=2$)
	distinguishes $G_1$ and $G_2$.
\end{lemma}
\begin{proof}
	Assume~$m$ iterations of $k$-walk refinement
	distinguish the two graphs
	and let~$\chi_i$ be the coloring obtained for $G_i$ for $i\in[2]$.
	Then there is a color $c$
	occurring for a different number of vertex pairs,
	say~$n_1$ and~$n_2$ with $n_1 > n_2$,
	in~$\chi_1$ and~$\chi_2$, respectively.
	Let $\phi(x,y)$ be the formula from Lemma~\ref{lem:formula-identifing-walk-refinement-color}
	of quantifier depth $m$
	that identifies vertex pairs of color $c$ in both colorings.
	Now the formula $\exists^{n_1} x,y.~\phi(x,y)$
	is of quantifier depth $m+1$ and
	distinguishes the graphs.
	
	Suppose now that $k=2$ (and hence the prior formula is
	not a valid $\walklogic{2}$ formula).
	Let~$D$ be the multiset of $c$\nobreakdash-outdegrees
	of all vertices in~$G_1$.
	Then the sum of all $c$\nobreakdash -outdegrees (respecting the multiplicity)
	is~$n_1$.
	Let $\mathcal{D}$ be the set of all possible $c$\nobreakdash-outdegree multisets
	with sum~$n_1$.
	Then the formula
	\[\bigvee_{D \in \mathcal{D}} \bigand_{(i,d) \in D} 
	\exists^i x.~\exists^d y.~\phi(x,y)\]
	distinguishes $G_1$ and $G_2$,
	where $(i,d) \in D$ says that~$d$ occurs with multiplicity~$i$ in~$D$.
\end{proof}

\section{Bijective Walk Pebble Game}
We now describe a game called the
\emph{bijective $k$-walk pebble game},
which corresponds to $k$-walk refinement and
$k$-walk counting logic.
It is an adaption of the bijective $3$-pebble game
to agree with the $k$-walk refinement.
 
There are two players, Spoiler and Duplicator.
The game is played on two undirected graphs 
${G_1=(V_1, E_1)}$ and ${G_2=(V_2, E_2)}$.
Spoiler obtains $k+1$ pairs of pebbles 
$(p_1, q_1), \dots, (p_{k+1},q_{k+1})$
labeled with numbers $1$ to $k+1$.
We say that the pebble pairs
$(p_i, q_i)$ and $(p_{i+1}, q_{i+1})$ for all $i \in [k]$
and the pairs $(p_{k+1}, q_{k+1})$ and $(p_1, q_1)$
are \emph{consecutive}. 
Given a pebble pair $(p_i, q_i)$
we will for simplicity write $(p_{i+1},q_{i+1})$
for the next consecutive pebble pair,
in particular in the case $i=k+1$, where $(p_1, q_1)$ is meant.

If $|V_1| \neq |V_2|$, Spoiler wins immediately.
Otherwise, all pebbles are placed beside the graphs. 
The game is played in multiple rounds.
One round consists of the following three moves:

\begin{enumerate}
	
	\item If there are pebbles already placed on the graphs,
	Spoiler can choose a pair of pebbles $(p_i, q_i)$,
	replace it with $(p_1, q_1)$ and then must also
	replace the next pair $(p_{i+1}, q_{i+1})$,
	if it is placed on the graph,
	with $(p_{k+1}, q_{k+1})$.
	In either case, she~(Spoiler) then picks up all pebble pairs apart the first and last.
		
	\item 
	Duplicator chooses a bijection $f\colon  V_1^{k-1} \to V_2^{k-1}$. 
	
	\item Spoiler places the pebbles $p_i$ for $2 \leq i \leq k$ onto vertices of $G_1$.
	She may place multiple pebbles on the same vertex.
	Assume pebble $p_i$ is placed onto vertex $u_i$
	and $f(u_2, \dots, u_{k}) = (v_2, \dots, v_{k})$.
	Then, Spoiler also places the pebbles $q_i$ onto $v_i$ for all~$2 \leq i \leq k$.

\end{enumerate}
Thus, as opposed to a bijection between vertices in the classic game, Duplicator chooses a bijection from the~$k$-walks in~$G_1$ from~$p_1$ to~$p_{k+1}$ to the~$k$-walks from~$q_1$ to~$q_{k+1}$ in~$G_2$ (hence the name walk pebble game). 

We say that Spoiler wins the game after the $i$-th round,
if there are consecutive pebble pairs
$(p_i, q_i)$ and $(p_{j+1}, q_{j+1})$
placed on vertices $(u, v) \in V_1^2$ and $(u',v') \in V_2^2$,
such that the induced subgraphs
$G_1[\{u,v\}]$ and $G_2[\{u',v'\}]$
are not isomorphic.
Duplicator wins the game if Spoiler never wins the game.

We say that Spoiler can \emph{force a win after the $i$-th round}
or has \emph{a winning strategy in~$i$~rounds},
if she can always win the game after the $i$-th round
for all possibles moves of Duplicator.

With \emph{bijective walk pebble game} we refer to the game
in which Spoiler is allowed to choose the number $k$ in the beginning
(after she has seen the two graphs).

Note that, similar to the toplevel quantifier of a $\walklogic{k}$ sentence,
in the first round, only $k-1$ pebbles are placed on the graph
and hence they describe a $(k-2)$\nobreakdash-walk.
Besides keeping definitions simpler,
in the case $k = 2$ this also ensures that the game becomes 
the  bijective 3-pebble game
(modulo some irrelevant replacements of pebble pairs).

In the following, let $G_1$ and $G_2$ be two undirected graphs.
Furthermore let $u,v \in V_1$, and $u', v' \in V_2$.
We say that the bijective $k$-walk pebble game
\emph{distinguishes} $(u,v)$
from $(u', v')$ in $m$ rounds,
if Spoiler has a winning strategy in~$m$ rounds in the game that has been altered as follows: instead of her making her first move 
in the first round, the pebble pairs $(p_1,q_1)$ and $(p_{k+1},q_{k+1})$
are placed on $(u,u')$ and $(v,v')$, respectively. Afterwards the game proceeds normally with Duplicator choosing a bijection and so on.
In the special case that $(u,v)$ is an edge, non-edge, or a loop but $(u',v')$ is not of the same type, we say that 
the vertex pairs are distinguished in 0 rounds.

\begin{lemma}
	\label{lem:walklogic-distinguish-edges-implies-pebble-game}
	Let $u,v \in V_1$ and $u',v' \in V_2$.
	If there is a $\walklogic{k}$ formula~$\phi(x,y)$
	of quantifier depth $m$
	that distinguishes $(u,v)$ from $(u', v')$,
	then so does the bijective $k$-walk pebble game
	in $m$ rounds.
\end{lemma}
\begin{proof}
	Assume that $\phi(x,y)$	distinguishes $(u,v)$ from $(u', v')$
	and that the pebble pairs $(p_1,q_1)$ and $(p_{k+1},q_{k+1})$ 
	are placed on these vertices.
	The proof proceeds by induction on~$m$.
	
	If $m=0$, $\phi(x,y)$ is quantifier free,
	hence $p_1$ and $p_{k+1}$ cover an edge, non-edge, or the same vertex,
	where $q_1$ and $q_{k+1}$ cover something different,
	and Spoiler wins the game immediately.

	Assume $\phi(z_1,z_{k+1})$ has quantifier depth $m+1$.
	If $\phi = \neg \phi'$,
	then $\phi'$ distinguishes $(u,v)$ from $(u', v')$, too.
	If $\phi = \phi_1 \wedge \phi_2$,
	one formula of~$\phi_1$ and~$\phi_2$ distinguishes~$(u,v)$ from $(u', v')$.
	Hence we can assume that $\phi$ is a walk quantifier: 
	\[\phi(z_1,z_{k+1}) = \exists^j z_2, \dots, z_{k}.~
    \bigwedge_{i \in [k]} \phi_{i}(z_{i}, z_{i+1}). \]
	Assume w.l.o.g.~that $\phi_{G_1}(u, v)$ is true but $ \phi_{G_2}(u', v')$ not.
	Duplicator chooses a bijection $f\colon  V_1^{k-1} \to V_2^{k-1}$.
	There must be a tuple$(w_2, \dots, w_{k}) \in V_1^{k-1}$
	serving as witness of the quantifier in $G_1$
	but $f(w_2, \dots, w_{k}) = (w'_2, \dots , w'_{k})$ does not serve as witness for~$G_2$,
	because otherwise $\phi_{G_2}(u', v')$ was true.
	Then Spoiler places for $2 \leq i \leq k$ the pebble~$p_i$ on~$w_i$
	and pebble~$q_i$ on $w'_i$.
	
	Now, there must be an $i \in [k] $
	such that $\phi_i(w_i, w_{i+1})$ is true
	but $\phi_i(w'_i, w'_{i+1})$ is not,
	because otherwise $(w'_2, \dots , w'_{k})$ is a witness.
	Now~$\phi_i$ is of quantifier depth~$m$ and
	the $i$-th and $(i+1)$-th pebble pairs are placed on the correct vertices.
	Thus Spoiler removes all other pebbles and wins the game
	in additional~$m$ rounds  by induction hypothesis.
\end{proof}

\begin{lemma}
	If there is a $\walklogic{k}$ sentence~$\phi$
	of quantifier depth~$m$ distinguishing~$G_1$ and~$G_2$,
	then Spoiler has a winning strategy in $m$ rounds
	in the bijective $k$-walk pebble game.
\end{lemma}
\begin{proof}
	Assume $\phi$ is a sentence of quantifier depth~${m>0}$
	distinguishing~$G_1$ and~$G_2$.
	For the same reasons as in Lemma~\ref{lem:walklogic-distinguish-edges-implies-pebble-game}, 
	we can assume that $\phi$ is a walk-quantifier:
	
	\[\phi = \exists^j z_1, \dots, z_{k-1}.
	~\bigwedge_{i \in [k-2]} \phi_i(z_i, z_{i+1}).\]
	Again as in Lemma~\ref{lem:walklogic-distinguish-edges-implies-pebble-game},
	for each bijection $f\colon  V_1^{k-1} \to V_2^{k-1}$
	there is a witness \linebreak[4]$(w_1, \dots, w_{k-1}) \in V_1^{k-1}$ of $G_1$
	such that $f(w'_1, \dots, w'_{k-1}) = (v_1, \dots, v_{k-1})$
	is not a witness of $G_2$.
	Again, there is a $j \in [k-2]$
	such that~$\phi_i$ distinguishes
	$(w_j, w_{j+1})$ and $(w'_j, w'_{j+1})$.
	
	When Spoiler places the~$p_i$ pebbles on the~$w_i$
	and the~$q_i$ pebbles on the~$w'_i$,
	Spoiler can force a win in additional
	${m-1}$ rounds by Lemma~\ref{lem:walklogic-distinguish-edges-implies-pebble-game}.
	So overall she has a winning strategy in $m$ rounds.
\end{proof}

\begin{lemma}
	\label{lem:pebble-game-disntguishes-edges-implies-walkrefinement}
	Let $u,v \in V_1$ and $u',v' \in V_2$.
	If the bijective $k$-walk pebble game distinguishes 
	$(u,v)$ from $(u',v')$ in $m$ rounds,
	then $m$ iterations of $k$-walk refinement
	distinguish them.
\end{lemma}
\begin{proof}
	Assume that the pebble pairs
	$(p_1,q_1)$ and $(p_{k+1},q_{k+1})$ are placed on $(u,v)$ and $(u',v')$
	and Spoiler has a winning strategy in additional $m$ rounds.
	Let $\chi$ and $\chi'$ be the initial colorings
	of the graphs $G_1$ and $G_2$.
	The proof proceeds by induction on $m$.
	
	If $m=0$, $(u,v)$ is an edge, non-edge, or loop,
	$(u',v')$ is not of the same type,
	and hence $\chi(u,v) \neq \chi'(u',v')$.
	
	Assume Spoiler can force a win of the game in additional $m+1$
	rounds.
	Whatever bijection Duplicator chooses,
	Spoiler can place the pebbles such that
	she can force a win in $m$ additional rounds.
	That means, by inductive hypothesis,
	that for every bijective mapping between
	the $k$\nobreakdash -walks from~$u$ to~$v$ in $\walkrefinementb{\chi}{k}^{m}$
	and the $k$\nobreakdash -walks from~$u'$ to~$v'$ in $\walkrefinementb{{(\chi')}}{k}^{m}$
	there is a walk
	that is mapped to a walk of different color.
	
	Hence, there is a $k$\nobreakdash -walk color 
	that occurs with different multiplicity from~$u$ to~$v$
	in~$\walkrefinementb{\chi}{k}^{m}$
	than from~$u'$ to~$v'$
	in $\walkrefinementb{{(\chi')}}{k}^{m}$.
	This just says that the vertex pairs
	obtain different colors in
	$\walkrefinementb{\chi}{k}^{m+1}$ respectively%
	$\walkrefinementb{{(\chi')}}{k}^{m+1}$.
\end{proof}

\begin{lemma}
	If Spoiler can force a win in the bijective $k$\nobreakdash-walk pebble game
	in $m$ rounds,
	then the $k$\nobreakdash-walk refinement distinguishes the graphs~$G_1$ and~$G_2$
	after $m$ iterations.
\end{lemma}
\begin{proof}
	At the beginning of the game, Duplicator chooses a bijection.
	For every such bijection,
	Spoiler can place the pebbles
	such that she can force a win in additional $m-1$ rounds.
	As in and by Lemma~\ref{lem:pebble-game-disntguishes-edges-implies-walkrefinement}
	there is no bijective mapping between the walks
	on vertices of $G_1$ and those on $G_2$
	such that assigned walks have the same color
	after $m-1$ iterations of $k$-walk refinement.
	
	But then $m$ iterations distinguish the graphs
	because if not, such a mapping would always exist.
\end{proof}

\begin{theorem}
	\label{thm:walk-refinement-logic-game-correspondence}
	Two graphs $G_1$ and $G_2$ are distinguished 
	by $k$-walk refinement if and only if
	they are distinguished by $\walklogic{k}$ if and only if
	Spoiler has a winning strategy in the bijective $k$-walk pebble game.
\end{theorem}
\begin{corollary}
	Two graphs $G_1$ and $G_2$ are distinguished 
	by walk refinement
	if and only if
	they are distinguished by walk counting logic
	if and only if
	Spoiler has a winning strategy in the bijective walk pebble game.
\end{corollary}

These equivalences in particular imply that the upper bound for the walk refinement (Theorem~\ref{thm:walk:ref:stab}) translates to the game and logic scenarios as follows.

\begin{corollary}
	If Spoiler has a winning strategy in the bijective walk-pebble game
	on two graphs $G_1$ and $G_2$,
	then she has a winning strategy requiring $\bigO(n)$ rounds.
\end{corollary}

\begin{corollary}
	If two graphs $G_1$ and $G_2$ are distinguished
	by walk counting logic,
	then they can be distinguished by a
	walk counting logic sentence of quantifier depth $\bigO(n)$.
\end{corollary}

\section{A Linear Lower Bound for Walk Refinement}
\label{sec:lower-bound}

In this section we show that there are graphs
on which walk refinement stabilizes only after $\Omega(n)$ iterations.
Specifically,
we show this for the same graphs, 
for which Fürer already showed that the WL refinement
requires $\Omega(n)$ iterations~\cite{Furer2001}.
We do this by demonstrating that Duplicator has a strategy
in the bijective walk pebble game played on these graphs
that delays the win of Spoiler for at least $\Omega(n)$ rounds.

In the following we recall well-known constructions and their properties
from~\cite{CaiFI1992} and~\cite{Furer2001}.  
For proofs of these properties we refer the reader to the original papers.

\subsection{CFI-Construction}
The graphs used by Fürer in~\cite{Furer2001}
to prove the lower bound
are obtained by taking suitable base graphs
and replacing each vertex with a special gadget.
We first describe these gadgets and their properties.

Let $G=(V,E)$ be a simple connected \emph{base graph}.
We call the vertices and edges in the base graph
\emph{base vertices} and \emph{base edges}, respectively. 
Each base vertex will be replaced by a gadget (a small graph)
and each base edge~$e$ will be represented by edges between the gadgets corresponding to the endpoints of~$e$.

Cai, Fürer, and Immerman introduced the so called CFI\nobreakdash-gadgets~\cite{CaiFI1992}
consisting of outer and middle vertices,
where a base edge results in edges between the outer vertices of two gadgets.
However, in~\cite{Furer2001} Fürer uses a variant of these gadgets
only consisting of the middle vertices.
He directly connects the middle vertices of two gadgets.
In this paper we follow this approach because
it simplifies our reasoning for the bijective walk pebble game (see Figures~\ref{fig:lower-bound-grid} and~\ref{fig:lower-bound-gadgets}).

A gadget~$F_d$ of degree~$d$ consists of all~$d$ tuples $\{0,1\}^d$
with an even number of ones as vertices. The gadget has no edges.

Let $v \in V$ be a base vertex of degree~$d$.
When replacing~$v$ with a gadget of degree~$d$,
we denote with $v(a_1, \dots , a_d)$, where $a_i \in \{0,1\}$, 
the vertices of the gadget (of course we have $a_i = 1$ for an even number of $a_i$).

We fix arbitrarily for each base vertex $v$ an ordering of its incident edges,
so that we can speak of the $i$-th base edge incident to $v$.
The undirected graph $\repl{G}$ is obtained from the base graph $G$
by replacing every base vertex $v$ with a gadget~$F(v)$ of degree~$d(v)$
and connecting gadgets arising from adjacent base vertices as follows:
Let $e= \{u,v\}$ be a base edge,
let~$u$ and~$v$ have degree~$d$ and~$d'$ respectively,
and assume that~$e$ is the $i$-th incident edge of~$u$
and the $j$-th incident edge of~$v$.
We then insert the edges 
$\setcondition{\{u(a_1, \dots, a_d), v(b_1, \dots, b_{d'})\}}{a_i = b_j}$
between vertices that agree on the $i$-th and $j$-th component, respectively.
That is, a base edge is represented in $\repl{G}$ by two
complete bipartite induced subgraphs,
one for $a_i = b_j = 0$
and one for $a_i = b_j = 1$
(these subgraphs may be empty if one of the both base vertices has degree~$1$).

To \emph{twist} a base edge $\{u,v\} \in E$ means to
replace every edge between a vertex of the gadget $F(u)$
and a vertex of the gadget $F(v)$
by a nonedge and vice versa.
We obtain from $\repl{G}$ another graph $\replt{G}$
by \emph{twisting} some arbitrary base edge.
The graph $\replt{G}$ is well defined up to isomorphism,
because the graph obtained from $\repl{G}$
by twisting another edge is always isomorphic to $\replt{G}$.
If there is at least one base edge,
then $\repl{G}$ is not isomorphic to $\replt{G}$.

We say that $v(a_1, \dots, a_d)$ \emph{originates} from $v$
or that the \emph{origin} of $v(a_1, \dots, a_d)$ is $v$. 
Likewise, we say that an edge
$\{u(a_1, \dots, a_d), v(b_1, \dots, b_{d'})\}$ \emph{originates} from $\{u,v\}$
or has \emph{origin} $\{u, v\}$.
We extend this notion to sets and walks: 
A set of vertices in $\repl{G}$ (or $\replt{G}$)
originates from the set of origins
and a walk originates from the walk consisting of the origins of the visited vertices.

In the following, we will use $\straight{u}, \straight{v},$ and $\straight{w}$
for vertices of $\repl{G}$
with origins $u, v$, and $w$ respectively.
Similarly, we use $\twisted{u}, \twisted{v},$ and $\twisted{w}$
for vertices of $\replt{G}$.
Let $\repl{G} = (\straight{V}, \straight{E})$ and $\replt{G} = (\twisted{V}, \twisted{E})$ and note that $\straight{V} = \twisted{V}$
and hence we can reinterpret an automorphism of $\replt{G}$
as a mapping between $\repl{G}$ and $\replt{G}$.
We say that an automorphism $\phi$ of $\replt{G}$
\emph{moves the twist} to the base edge $\{w, w'\} \in E$
if
\[
\{\straight{u}, \straight{v}\} \in \straight{E} 
\Leftrightarrow
\{\phi(\straight{u}), \phi(\straight{v})\} \in \phi(\twisted{E})
\]
for all $\{\straight{u},\straight{v}\}$ not originating from $\{w, w'\}$ and
\[
\{\straight{u}, \straight{v}\} \in \straight{E} 
\Leftrightarrow
\{\phi(\straight{u}), \phi(\straight{v})\} \notin \phi(\twisted{E})
\]
for all $\{\straight{u},\straight{v}\}$ originating from $\{w, w'\}$.

That is, $\phi$ behaves like an isomorphism 
except on vertex pairs originating from the base edge $\{w, w'\}$,
on which~$\phi$ inverts adjacency.
For every base edge, there is an automorphism of $\replt{G}$
moving the twist to that edge
(which is just another way of saying
that the graph obtained by twisting some edge in $\repl{G}$
is always isomorphic to $\replt{G}$).
Assume that~$\phi$ moves the twist to $\{u, v\}$
and we want to move the twist to $\{v, w\}$.
Then there is another automorphism~$\psi$
which possibly permutes the vertices of $F(v)$ but is otherwise constant
such that $\psi \circ \phi$ moves the twist to $\{v, w\}$.
In general, not every automorphism moves the twist to a single base edge,
but to an odd number of bases edges (on which it inverts adjacency).
In the following we only consider automorphisms
moving the twist to a single base edge.

\subsection{Lower Bound for the Weisfeiler-Leman Refinement}
We recall the necessary parts of Fürer's lower bound 
on the iteration number of the
$d$-dimensional WL refinement.
We only deal with the $2$-dimensional case.

\tikzstyle{vertex} = [circle, fill=black, inner sep=0.6mm,  minimum size = 1mm]
\tikzstyle{twisted} = [decorate, decoration=zigzag]

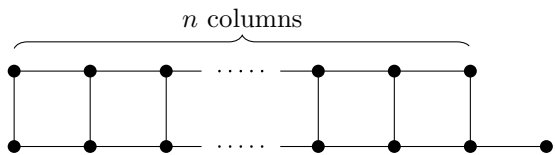
\begin{figure}
	\centering
		\begin{tikzpicture}[>=stealth,shorten >=1pt,node distance=1.5cm,on grid,every node/.style={vertex}]
		
		\foreach \x in {1,...,3}
		{
			\foreach \y in {1,...,2}
			{
				\node (\x\y) at (\x, \y){};
				\ifthenelse{\NOT 1 = \x} {
					\path[black] (\x,\y) edge (\x-1,\y);
				}{}
			}
			\path[black] (\x,1) edge (\x, 2);
		}

		\foreach \x in {5,...,7}
		{
			\foreach \y in {1,...,2}
			{
				\node (\x\y) at (\x, \y){};
				\ifthenelse{\NOT 5 = \x} {
					\path[black] (\x,\y) edge (\x-1,\y);
				}{}
			}
			\path[black] (\x,1) edge (\x, 2);
		}		
	
		\node (91) at (8,1){};
		
		\path[black] (7,1) edge (8,1);
		
		\foreach \y in {1,...,2}
		{
			\path[black, line width = 1pt, line cap=round, dash pattern=on 0pt off 0.14cm] (3.7,\y) edge (4.3,\y);
			\path[black]
				(3,\y) edge (3.5,\y)
				(4.5,\y) edge (5,\y);
		}	
	
		\draw [decorate,decoration={brace,amplitude=5pt},xshift=0pt,yshift=0pt]
		(1,2.3) -- (7,2.3) node [fill=none, black,midway,yshift=0.4cm] 
		{\footnotesize $n$ columns};

		\end{tikzpicture}
	\caption{
		The base graph $\grid{2}{n}$.
	}
	\label{fig:lower-bound-grid}
\end{figure}

Let $\grid{2}{n}$ be a $2\times n$ grid with an additional vertex attached to one corner
(depicted in Figure~\ref{fig:lower-bound-grid}).
In this graph all vertices can be uniquely identified by their distance to
the unique vertex of degree $1$
as well as the distance to the two adjacent vertices of degree $2$
(given that $n \geq 3$).
By replacing with gadgets as described, we obtain two graphs
$\repl{\grid{2}{n}}$ and $\replt{\grid{2}{n}}$.
The graph $\repl{\grid{2}{4}}$ is shown in Figure~\ref{fig:lower-bound-gadgets}.
The graphs $\repl{\grid{2}{n}}$ and $\replt{\grid{2}{n}}$
are not isomorphic
and can be distinguished by the Weisfeiler-Leman refinement.
Hence, Spoiler has a winning strategy in the bijective 3-pebble game
and consequently also in the bijective walk pebble game.
Note that, since vertices of $\grid{2}{n}$ have degree at most $3$,
$\repl{\grid{2}{n}}$ and $\replt{\grid{2}{n}}$ have $\Theta(n)$ vertices.

We recall some facts for the bijective 3-pebble game played on the graphs from~\cite{Furer2001}:
Assume that in the progress of the game
some pebble pairs are placed on the graphs.
Then we call an isomorphism $\phi$ between two graphs \emph{pebble respecting}
if~$v$ is covered by pebble~$p_i$ if and only if $\phi(v)$ is covered by~$q_i$.
In case~$\phi$ is an automorphism it necessarily maps all vertices covered by pebbles to themselves.

\newcommand{\id}[1]{#1}

% #1 degree
% #2 name
% #3 x position
% #4 y position
\newcommand{\drawgadget}[4] {
	
	\draw[dashed] (#3,#4) circle (0.6);

	\ifnum #1=1
		\node (#2\id{0}) at (#3,#4){};
	\else \ifnum #1=2
		\node (#2\id{11}) at (#3+0.25,#4-0.25){};
		\node (#2\id{00}) at (#3-0.25,#4+0.25){};
	
	\else
		\foreach \x in {0,...,1}
		\foreach \y in {0,...,1}
		\foreach \z in {0,...,1}
		{
			\ifodd \numexpr \x + \y + \z \else
				\node (#2\x\y\z) at ({#3 + 0.5*( \z - 0.3*\x) -0.2}, {#4 + 0.5*(\x + 0.7*\z - 1.7*\y -\x*\z)+0.1}){}; %[label={below=0pt:$(\x,\y,\z)$}]
			\fi
		}
	\fi
	\fi

}

\newcounter{memcount}
\newcommand{\countmem}[2]{%
	\setcounter{memcount}{0}%
	\foreach \i in #1{%
		\ifthenelse{\equal{\i}{#2}} {%
		\stepcounter{memcount}%
	}{}%
	}%

}

\newcommand{\suffx}{x}
\newcommand{\suffy}{y}
\newcommand{\suffz}{z}
\newcommand{\suffsum}{sum}
\newcommand{\suffseq}{seq}
\newcommand{\suffrr}{rr}

\newcommand{\foreachdim}[4] {
	\ifnum #1=1
		\expandafter\foreach \csname#2\suffx\endcsname in {0,...,1} {
			\expandafter\def\csname#2\suffsum\endcsname{%
					\csname#2\suffx\endcsname}
			\expandafter\def\csname#2\suffseq\endcsname{%
					\csname#2\suffx\endcsname}
			\expandafter\def\csname#2\suffrr\endcsname{%
					{\csname#2\suffx\endcsname}}
			#3
		}
	\else \ifnum #1=2
		\expandafter\foreach \csname#2\suffx\endcsname in {0,...,1} {
		\expandafter\foreach \csname#2\suffy\endcsname in {0,...,1} {
			\expandafter\def\csname#2\suffsum\endcsname{%
				\csname#2\suffx\endcsname + %
				\csname#2\suffy\endcsname}
			\expandafter\def\csname#2\suffseq\endcsname{%
				\csname#2\suffx\endcsname\csname#2\suffy\endcsname}
			\expandafter\def\csname#2\suffrr\endcsname{%
				{\csname#2\suffx\endcsname,\csname#2\suffy\endcsname}}
				#3
		}
		}
	\else
		\expandafter\foreach \csname#2\suffx\endcsname in {0,...,1} {
		\expandafter\foreach \csname#2\suffy\endcsname in {0,...,1} {
		\expandafter\foreach \csname#2\suffz\endcsname in {0,...,1} {
		\expandafter\def\csname#2\suffsum\endcsname{%
			\csname#2\suffx\endcsname + %
			\csname#2\suffy\endcsname +%
			\csname#2\suffz\endcsname}
		\expandafter\def\csname#2\suffseq\endcsname{%
			\csname#2\suffx\endcsname\csname#2\suffy\endcsname\csname#2\suffz\endcsname}
		\expandafter\def\csname#2\suffrr\endcsname{%
			{\csname#2\suffx\endcsname,\csname#2\suffy\endcsname,\csname#2\suffz\endcsname}}
		#3
		}
		}
	}
	\fi\fi
	
}

% #1 dim 1
% #2 dim 2
% #3 name 1
% #4 name 2
% #5 index 1
% #6 index 2
% #7 comma separated list of edges to bend right (format of an edge: indices1'E'indices2, e.g. 011E110)
% #8 comma separated list of edges to bend left
\newcommand{\drawconnectgadget}[8] {

	\foreachdim{#1}{a} {
	\foreachdim{#2}{b}{
		\pgfmathsetmacro\as{\arr[#5]}
		\pgfmathsetmacro\bs{\brr[#6]}
		\ifodd \numexpr \asum \else
		\ifodd \numexpr \bsum \else
		\if\as\bs
		    \countmem{{#7}}{\aseq\id{E}\bseq}
		    \ifthenelse{\value{memcount} > 0}{
		    	\path (#3\aseq) edge[bend right=5] (#4\bseq);
		    }{
		    	\countmem{{#8}}{\aseq\id{E}\bseq}
		    	\ifthenelse{\value{memcount} > 0}{
		    		\path (#3\aseq) edge[bend left=5] (#4\bseq);
		    	}{
		    		\path (#3\aseq) edge (#4\bseq);
		    	}
		    }

		\fi
		\fi\fi
	}
}
	
}

%#1 = name
%2-4: indices
%#5 = offset
\newcommand{\gadgetlabel}[5] {
	\node[fill = none] () at ($(#1#2#3#4) + #5$) {\scriptsize$#1(#2,#3,#4)$};
}

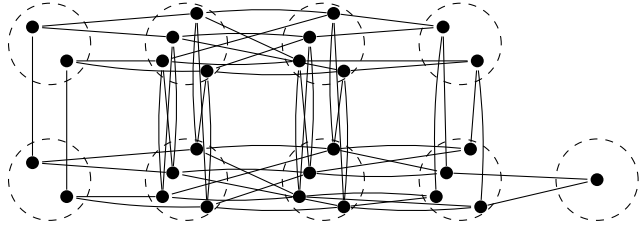
\begin{figure}
	\centering
		\begin{tikzpicture}[>=stealth,shorten >=1pt,on grid,every node/.style={vertex}, scale=0.9]
		
		\def\spacing{2}

		\drawgadget{2}{z1}{-\spacing}{0}
		\drawgadget{2}{z2}{-\spacing}{\spacing}
		
		\drawgadget{3}{a1}{0}{0}
		\drawgadget{3}{b1}{\spacing}{0}
		\drawgadget{3}{c1}{2*\spacing}{0}
		\drawgadget{3}{a2}{0}{\spacing}
		\drawgadget{3}{b2}{\spacing}{\spacing}
		\drawgadget{2}{c2}{2*\spacing}{\spacing}
		
		\drawconnectgadget{3}{3}{b1}{a1}{0}{0}{101E101,000E000}{011E011, 110E110}
		\drawconnectgadget{3}{3}{b1}{c1}{1}{1}{000E000,011E011}{101E101, 110E110}
		\drawconnectgadget{3}{3}{b2}{a2}{0}{0}{101E101,000E000}{110E110,011E011}
		\drawconnectgadget{3}{2}{b2}{c2}{1}{1}{}{}
		
		\drawgadget{1}{d}{3*\spacing}{0}
		\drawconnectgadget{3}{3}{a1}{a2}{2}{2}{000E000,011E011}{101E101,110E110}
		\drawconnectgadget{3}{3}{b1}{b2}{2}{2}{000E000,011E011}{101E101,110E110}
		\drawconnectgadget{3}{2}{c1}{c2}{2}{0}{011E11}{110E00}
		\drawconnectgadget{3}{1}{c1}{d}{0}{0}{}{}
		
		\drawconnectgadget{3}{2}{a1}{z1}{1}{0}{}{011E11}
		\drawconnectgadget{3}{2}{a2}{z2}{1}{0}{}{011E11}
		\drawconnectgadget{2}{2}{z1}{z2}{1}{1}{}{}
		
		\end{tikzpicture}
	\caption{
		The graph $\grid{2}{4}$.
		A circle indicates a base vertex
		and hence contains one gadget.
	}
	\label{fig:lower-bound-gadgets}
\end{figure}

Intuitively,
since Duplicator can move the twist to different edges using automorphisms,
Spoiler needs to ``catch'' the twist with her pebbles.
If a vertex is placed on a vertex $\twisted{v}$
of $\replt{\grid{2}{n}}$ with origin $v$,
all pebble respecting automorphisms of $\replt{\grid{2}{n}}$ fix all vertices in $F(v)$.
Hence, it does not matter
on which vertex originating from $v$ Spoiler places a pebble
and we can simply say that Spoiler places a pebble on $v$.

A set of vertices of $\repl{\grid{2}{n}}$ (or $\replt{\grid{2}{n}}$)
is called a \emph{wall}
if its origin is a separator of~$\grid{2}{n}$,
that is, a set of vertices whose
removal separates the graph into at least two connected components.
We say that Spoiler \emph{builds a wall},
if the vertices covered by the pebbles form a wall.
To avoid a quick win for Spoiler,
Duplicator picks the bijection on her turn so that it is 
origin respecting.
That is, Duplicator maps a vertex~$\straight{u}$ to a vertex~$\twisted{u}$
with the same origin.
Our strategy for Duplicator
in the bijective walk-pebble game described below has this property, too.
Consequently, when asking whether the pebbles form a wall or not,
it does not matter whether we consider the pebbles on
$\repl{\grid{2}{n}}$ or $\replt{\grid{2}{n}}$). 
Since we will only consider origin respecting strategies,
we will often simply think of the origins being pebbled.

Suppose now some vertices of $\grid{2}{n}$ have been pebbled.
A~\emph{component} of the graph $\grid{2}{n}$, w.r.t.~the pebbled vertices,
is an inclusion-wise maximal and nonempty set of base edges~$C$ satisfying the following property:
For every two edges $e_1, e_2 \in C$,
there is a walk $(v_1, \dots , v_j)$ only using edges
$\{v_i, v_{i+1}\} \in C$ for all $i \in [j-1]$
such that ${e_1 = \{v_1, v_2\}}$, ${e_2 = \{v_{j-1}, v_j\}}$,
and $v_2, \dots, v_{j-1}$ are not covered by a pebble.
A component $C$ \emph{contains} a base vertex $v$,
if all base edges incident to $v$ are contained in $C$.
The \emph{size} of the component is
the number of vertices it contains.
We call a component \emph{nontrivial}, when its size is nonzero.

Intuitively, one can think of components as the parts of the graph $\grid{2}{n}$
obtained by deleting only the vertices covered by pebbles,
but not the edges incident to these vertices.
This results in ``dangling'' edges in nontrivial components
(edges, which were incident to only one vertex covered by a pebble)
and edges not incident to any vertex forming the trivial components
(edges, both endpoints of which where covered by a pebble).

We call a component $C$ \emph{twisted},
if there is a pebble respecting automorphism that moves the twist to an edge in $C$.
If $C$ is twisted, every pebble respecting automorphism
moves the twist to an edge in $C$, i.e.,
a twisted component contains precisely the edges,
to which Duplicator can move the twist
with pebble respecting automorphisms.

With $3$ pebbles Spoiler can build at most one wall
and hence in the bijective $3$-pebble game
there are at most two nontrivial components.
When a trivial component is twisted, Spoiler wins the game.
To delay the win of Spoiler,
Duplicator maintains a single twisted component,
whose size only decreases by a constant per round.

\subsection{Lower Bound for Walk Refinement}

The situation changes in the bijective walk pebble game, since the game does not have a bound on the number of pebbles that are used.
In the situation where more than two pebbles are placed on the graph,
there can be many components
(Spoiler may in particular cover every vertex and every edge).
But, once she has to remove all but two pebbles,
there can be at most one wall again.
We describe a strategy of Duplicator with the following properties:
\begin{enumerate}
	\item If the size of the twisted component is at most ${n-2}$,
	its size reduces by at most $2$ after one round.
	\item If the size of the twisted component is greater than ${n-2}$ (e.g.~in the beginning of the game),
	the size of the twisted component is at least $n-4$ after one round.
\end{enumerate}
Combining these properties, Spoiler needs at least $\Omega(n)$ rounds to win.
Intuitively, the existence of such a strategy
comes from the fact that in the bijective walk pebble game
Duplicator needs to preserve adjacency and/or equality only for consecutive pebble pairs.
This allows her to fix the twist locally, possibly introducing global inconsistencies 
but never introducing inconsistencies between consecutive pebble pairs.
We describe a strategy how Duplicator
can introduce these inconsistencies
only on edges incident to a chosen base vertex.

\newcommand{\bij}[3]{f_{{#2},{#3}}^{#1}}

Suppose we are in the bijective $k$-walk pebble game for an arbitrary $k$
and we are in the state of the game in which only two pebble pairs are placed on the graphs.
Assume that the pebbles are placed on
$\straight{u}_1$ and $\straight{u}_2$ in $\repl{\grid{2}{n}}$ and on
$\twisted{u}_1$ and $\twisted{u}_2$ in $\replt{\grid{2}{n}}$ respectively.
We assume, justifying our notation,
that $\straight{u}_i$ and $\twisted{u}_i$ have the same origin for $i \in \{1,2\}$
and that Spoiler has not won the game already,~i.e.,~%
the pebbles define an isomorphism between the subgraphs
induced by the pebbled vertices.

Let $v$ be a base vertex of degree $3$ in the twisted component
and $e_1$ and $e_2$ be two distinct base edges incident to $v$.
We define a bijection
$\bij{v}{e_1}{e_2} \colon \straight{V}^{k-1} \to \twisted{V}^{k-1}$
as follows:

First, pick a pebble respecting isomorphism $\phi$ of $\replt{\grid{2}{n}}$
moving the twist to $e_1$ (it exists because $v$ is in the twisted component).
Let $(\straight{w}_1, \dots , \straight{w}_{k-1}) \in \straight{V}^{k-1}$
and set ${\straight{w}_0 := \straight{u}_1}$ and
${\straight{w}_{k} := \straight{u}_2}$.
We define $\bij{v}{e_1}{e_2}(\straight{w}_1, \dots , \straight{w}_{k-1}) = (\tilde{w}_1, \dots, \tilde{w}_{k-1})$ entry-wise for $i \in [k-1]$
by the following case distinction:
\begin{enumerate}
	\item If $w_i \neq v$, we set $\tilde{w}_i := \phi(\straight{w}_i)$.
	\item If $w_i = v$, let $j < i$ and $\ell>i$ be the unique indices 
	such that
	$$w_j \neq w_{j+1} = \dots = w_i = \dots = w_{\ell-1} \neq w_{\ell}.$$
	We pick an edge $e$ incident to $v$ by a second case distinction:
	\begin{enumerate}
		\item If $e_1:=\{w_j, w_i\}$ and $e_2:=\{w_i, w_{\ell}\}$ are distinct base edges,
		let $e \notin \{e_1, e_2\}$ be the third base edge
		incident to $v$.
	
		\item
		Otherwise at least one of
		$\{w_j, w_i\}$ and $\{w_i, w_{\ell}\}$
		is not a base edge or \linebreak$\{w_j, w_i\} = \{w_i, w_{\ell}\}$.
		Hence, when passing through $v$ at position $i$,
		the walk uses at most one base edge $e'$ incident to $v$.
		Let the edge $e \in \{e_1, e_2\} \setminus \{e'\}$
		be smallest one according to the fixed order of base edges incident with~$v$.
		
	\end{enumerate}

	Finally, let $\psi$ be the automorphism of $\replt{\grid{2}{n}}$
	such that $\psi \circ \phi$ moves the twist to $e$
	and $\psi$is the identity on all vertices
	apart from those in the gadget $F(v) = F(w_i)$.
	We set ${\tilde{w}_i := \psi(\phi(\straight{w}_i))}$.

\end{enumerate}

\begin{lemma}
	The function $\bij{v}{e_1}{e_2}$ is a bijection.
\end{lemma}
\begin{proof}
	The function $\bij{v}{e_1}{e_2}$ maps a walk in $\repl{\grid{2}{n}}$
	to a walk in $\replt{\grid{2}{n}}$ with the same origin.
	For two walks in $\repl{\grid{2}{n}}$ with the same origin,
	Duplicator chooses for each of the walks
	the same additional automorphisms~$\psi$.
	Since~$\phi$ and all chosen~$\psi$ are bijections,
	the map $\bij{v}{e_1}{e_2}$ is a bijection.
\end{proof}

\begin{lemma}
	\label{lem:bij-not-loose}
	If Duplicator chooses $\bij{v}{e_1}{e_2}$ as bijection,
	then Spoiler does not win in the current round.
\end{lemma}
\begin{proof}
	Assume Spoiler picks some $(\hspace{-1pt}\straight{w}_1, \dots, \straight{w}_{k-1}\hspace{-1pt})$ and
	$\bij{v}{e_1}{e_2}(\hspace{-1pt}\straight{w}_1, \dots, \straight{w}_{k-1}\hspace{-1pt}) \hspace{-1pt}=\hspace{-1pt} (\hspace{-1pt}\tilde{w}_1, \dots, \tilde{w}_{k-1}\hspace{-1pt})$.
	Then the $p_i$ pebbles are placed on the $\straight{w}_i$
	and the $q_i$ pebbles are placed on the $\tilde{w}_i$.
	We set $\straight{w}_0 := \straight{u}_1$, $\tilde{w}_0 := \phi(\straight{u}_1)$
	and likewise $\straight{w}_k := \straight{u}_2$ and $\tilde{w}_k = \phi(\straight{u}_2)$.
	Note that $\straight{w}_0$ and $\tilde{w}_0$ are covered by the first pebble pair
	and $\straight{w}_k$ and $\tilde{w}_k$ by the last one
	because $\phi$ was chosen pebble respecting.
	Suppose~$i \in \{0, \dots , k-1\}$.
	To show that Spoiler does not win in this round,
	we show that the $i$\nobreakdash-th and $(i+1)$\nobreakdash-th pebble pair
	define an isomorphism
	between the subgraphs
	of $\repl{\grid{2}{n}}$ and $\replt{\grid{2}{n}}$
	induced by the vertices covered 
	by the $i$-th and $(i+1)$-th pebble pair.
	
	\begin{itemize}
		\item 
		Suppose $\straight{w}_i$ and $\straight{w}_{i+1}$ do no originate from $v$
		and thus $\twisted{w}_i = \phi(\straight{w}_i)$ and\linebreak
		$\twisted{w}_{i+1} = \phi(\straight{w}_{i+1})$.
		Then $\{\straight{w}_i, \straight{w}_{i+1}\}$ in particular
		does not originate from $e_1$
		and hence by the choice of $\phi$ the pebbles define an isomorphism
		of the induced subgraphs.
		
		\item
		Suppose $\straight{w}_i$ and $\straight{w}_{i+1}$ both originate from $v$.
		In this case $\twisted{w}_i = \psi(\phi(\straight{w}_i))$ and
		${\twisted{w}_{i+1} = \psi(\phi(\straight{w}_{i+1}))}$
		for some automorphism $\psi$.
		Because both $\phi$ and $\psi$ are bijections and 
		inside a gadget there are no edges, the pebbles define an isomorphism.
		
		\item
		Lastly, suppose $\straight{w}_i$ originates from $v$ but
		$\straight{w}_{i+1}$ does not.
		In this case ${\twisted{w}_i = \psi(\phi(\straight{w}_i))}$ 
		and ${\twisted{w}_{i+1} = \phi(\straight{w}_{i+1})}$
		for another automorphism $\psi$ such that ${\psi \circ \phi}$
		moves the twist to an edge other than  $\{v, w_{i+1}\}$.
		Therefore ${\{\straight{w}_i, \straight{w}_{i+1}\} \in \straight{E}}$ holds if and only if ${\{\psi(\phi(\straight{w}_i)), \psi(\phi(\straight{w}_{i+1}))\} = \{\twisted{w}_i, \twisted{w}_{i+1}\} \in \phi(\twisted{E})}$.
		Recall that $\psi$ was chosen constant on all vertices apart $F(v)$
		and thus $\twisted{w}_{i+1} = \phi(\straight{w}_{i+1}) = \psi(\phi(\straight{w}_{i+1}))$.
		The case where~$\straight{w}_{i+1}$ originates from~$v$
		but~$\straight{w}_i$ does not is symmetric.
		\qedhere
	\end{itemize}	
\end{proof}

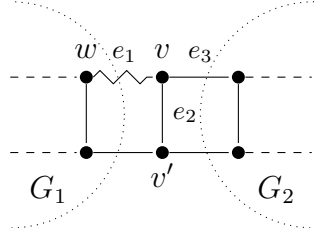
\begin{figure}
	\centering
		\begin{tikzpicture}[>=stealth,shorten >=1pt,node distance=1.5cm,on grid,every node/.style={vertex}]
		\node[fill=none] () at (0,0.35) {$\strut w$};
		\node (U1) at (0,0){};
		\node (U2) at (0,-1){};
		
		\node[fill=none] () at (1,0.35) {$\strut v$};
		\node (V1) at (1,0){};
		\node[fill=none] () at (1,-1.35) {$\strut v'$};
		\node (V2) at (1,-1){};
		
		\node (P3) at (2,0){};
		\node (P4) at (2,-1){};
		
		\node[fill=none] at (0.5,0.3) {\footnotesize$e_1$};
		\node[fill=none] at (1.3,-0.5) {\footnotesize$e_2$};
		\node[fill=none] at (1.5,0.3) {\footnotesize$e_3$};
		
		\path
		(U1) edge (U2)
		(U1) edge[twisted] (V1)
		(U2) edge (V2)
		(V1) edge (V2)
		(V1) edge (P3)
		(V2) edge (P4)
		(P3) edge (P4);
		\draw[dashed]
		(-1,0) edge (U1)
		(-1,-1) edge (U2)
		(P3) edge (3,0)
		(P4) edge (3,-1);
		
		\begin{scope}
		\clip (-1,1) rectangle (3,-2);
		\draw[dotted] (-1.0,-0.5) circle(1.5);
		\draw[dotted] (3.0,-0.5) circle(1.5);
		\end{scope}
		
		\node[fill=none] at (-0.5,-1.5) {$\strut G_1$};
		\node[fill=none] at (2.5,-1.5) {$\strut G_2$};
		
		\end{tikzpicture}
	\caption{
		The situation in Lemma~\ref{lem:remove-pebbles-component-size}:
		dashed edges may exist but not have to
		and the two subgraphs $G_1$ and $G_2$ are indicated by dotted lines.
		The isomorphism $\phi$ moves the twist to $e_1$.
	}
	\label{fig:bij-separator}
\end{figure}

\begin{lemma}
	\label{lem:remove-pebbles-component-size}
	Let $e_1 = \{v,w\}$ and $e_2 = \{v,v'\}$ be base edges,
	$v' \neq w$, and $\{v,v'\}$ be a separator of $\grid{2}{n{}}$
	separating $\grid{2}{n}$ into two subgraphs $G_1$ and $G_2$.
	Assume that $w$ is contained in $G_1$ and that $G_2$ has $\ell$ vertices.
	
	If Duplicator chooses $\bij{v}{e_1}{e_2}$ as bijection, Spoiler pebbles two corresponding walks and then removes all pebble pairs apart from two consecutive ones,
	then the new twisted component has size at least $\min\{\ell, 2n-\ell-2\}$.
\end{lemma}
\begin{proof}
	The situation of the lemma is shown in Figure~\ref{fig:bij-separator}.
	If $G_2$ contains $\ell$ vertices, then $G_1$ contains $\ell' := 2n-\ell-1$ vertices.
	We first note that if there is no wall,
	the twisted component has size $2n-1$.
	In the case of a wall, we make the following case distinction:
	\begin{itemize}
		\item The vertex~$v$ is not covered by a pebble.
		Then $\phi$ is still a pebble respecting automorphism moving the twist to~$e_1$ and thus~$v$ is in the twisted component.
		It has size at least $\min\{\ell, \ell'\} + 1$.
		
		\item The vertex $v$ and an adjacent vertex $u$ are covered by pebbles.
		If $\{v,u\} \in \{e_1, e_3\}$ then there is no wall.
		So $u = v'$ and thus the twist can be moved to $e_1$ or $e_3$
		by the choice of $\psi$ in the construction of $\bij{v}{e_1}{e_2}$.
		Then the twisted component has size $\min\{\ell, \ell'\}$.
		
		\item The vertex $v$ and a non-adjacent vertex $u$ are covered by pebbles.
		In the construction of $\bij{v}{e_1}{e_2}$  the automorphism~$\psi$
		is chosen such that the twist is moved to~$e_1$ or~$e_2$.
		If~$u$ is in~$G_1$, the twisted component has size $\ell +1$ 
		if the twist was moved to~$e_2$
		and size $\ell' - 1$ is the twist was moved to~$e_1$.
		Otherwise~$u$ is in~$G_2$. Since the twist can be moved to~$e_1$ or~$e_2$,
		the twisted component has size $\ell + 1$.
		\qedhere
	\end{itemize}

\end{proof}

\tikzstyle{vertex} = [circle, fill=black, inner sep=0.6mm,  minimum size = 1mm]
\tikzstyle{twisted} = [decorate, decoration=zigzag]

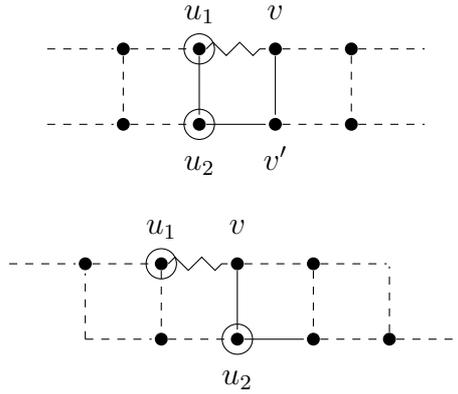
\begin{figure}
	\centering
		\begin{tikzpicture}[>=stealth,shorten >=1pt,node distance=1.5cm,on grid,every node/.style={vertex}]
			\node[fill=none] () at (0,0.5) {$\strut u_1$};
			\node (U1) at (0,0){};
			\node[fill=none] () at (0,-1.5) {$\strut u_2$};
			\node (U2) at (0,-1){};

			\node[fill=none] () at (1,0.5) {$\strut v$};
			\node (V1) at (1,0){};
			\node[fill=none] () at (1,-1.5) {$\strut v'$};
			\node (V2) at (1,-1){};
			
			\node (P1) at (-1,0){};
			\node (P2) at (-1,-1){};
			\node (P3) at (2,0){};
			\node (P4) at (2,-1){};
						
			\path
			  (U1) edge (U2)
			  (U1) edge[twisted] (V1)
			  (U2) edge (V2)
			  (V1) edge (V2);
			\path[dashed]
			  (V1) edge (P3)
			  (V2) edge (P4)
			  (P3) edge (P4)
			  (P1) edge (U1)
			  (P1) edge (P2)
			  (P2) edge (U2);
			\draw[dashed]
			   (-2,0) edge (P1)
			   (-2,-1) edge (P2)
			   (P3) edge (3,0)
			   (P4) edge (3,-1);
			\draw
			   (U1) circle (0.2cm)
			   (U2) circle (0.2cm);
		\end{tikzpicture}\\
		\begin{tikzpicture}[>=stealth,shorten >=1pt,node distance=1.5cm,on grid,every node/.style={vertex}]
		\node[fill=none] () at (0, 0.5) {$\strut u_1$};
		\node (U1) at (0,0){};
		\node[fill=none] () at (1,-1.5) {$\strut u_2$};
		\node (U2) at (1,-1){};
		
		\node[fill=none] () at (1, 0.5) {$\strut v$};
		\node (V1) at (1,0){};
		\node (V2) at (2,-1){};
		
		\node (P1) at (-1,0){};
		\node (P2) at (0,-1){};
		\node (P3) at (2,0){};
		\node (P4) at (3,-1){};
		
		\path
		(U1) edge[twisted] (V1)
		(V1) edge (U2)
		(U2) edge (V2);
		\path[dashed]
		(P1) edge (U1)
		(U1) edge (P2)
		(P2) edge (U2);
		\draw[dashed]
		(V1) edge (P3)
		(P3) edge (V2)
		(V2) edge (P4)
		(-2,0) edge (P1)
		(-1,-1) edge (P2)
		(-1,-1) edge (P1)
		(P3) edge (3,0)
		(P4) edge (4,-1)
		(P4) edge (3,0);
		\draw
		(U1) circle (0.2cm)
		(U2) circle (0.2cm);
		
		\end{tikzpicture}
	\caption{
		The two possible situations in Theorem~\ref{thm:lower-bound-game}:
		the base vertices $u_1$ and $u_2$ are covered by pebbles
		(indicated by circles) and form a wall.
		Dashed edges can exists, but do not have to.
		The automorphism $\phi$ moves the twist to $\{u_1, v\}$
		and the twisted component is to the right of the wall.
	}
	\label{fig:wall-lower-bound-strategy}
\end{figure}

\begin{theorem}
	\label{thm:lower-bound-game}
	For every $k\geq 2$ and $n\geq 3$ Duplicator has a strategy
	in the bijective $k$-walk pebble game
	played on the graphs 
	$\repl{\grid{2}{n}}$ and $\replt{\grid{2}{n}}$
	such that Spoiler wins in $\Omega(n)$ rounds at the earliest.
\end{theorem}
\begin{proof}
	We show that Duplicator has a strategy
	such that after $m$ rounds the twisted component of $\grid{2}{n}$
	is of size at least $n - 2(m+1)$
	and Spoiler can win only if its size is at most~$2$. 
	
	Assume that after $m$ rounds the twisted component has size
	at least $n - 2(m+1) > 2$,
	that there are only two pebbles on the graphs,
	and its Duplicator's turn to pick a bijection.
	There are two cases:
	\begin{itemize}
		\item The twisted component
		has size at most $n-2$.
		Then in particular Spoiler builds a wall.
		Let the pebbles be placed on base vertices~$u_1$ and~$u_2$.
			Duplicator picks~$v$ as depicted in Figure~\ref{fig:wall-lower-bound-strategy}:
			If~$u_1$ and~$u_2$ are adjacent,~%
			$v$ is the neighbor of~$u_1$ in the twisted component
			and~$v'$ is the common neighbor of~$v$ and~$u_2$.
			Duplicator chooses
			$\bij{v}{\{u_1, v\}}{\{v,v'\}}$ as bijection.
			If otherwise~$u_1$ and~$u_2$ are not adjacent,~%
			$v$ is the neighbor of both~$u_1$ and~$u_2$ in the twisted component.
			We possibly exchange names of~$u_1$ and~$u_2$
			so that $\{u_2, v\}$ is a separator of the graph.
			Then Duplicator chooses
			$\bij{v}{\{u_1, v\}}{\{u_2,v\}}$ as bijection.
			By Lemma~\ref{lem:bij-not-loose}, Spoiler does not win in
			the current round and
			by Lemma~\ref{lem:remove-pebbles-component-size}
			the size of the new twisted component is at least
			$n - 2(m+2)$.
		\item The twisted component has size greater than $n-2$.
		Let $\{u_1, u_2\}$ be a base edge and a separator of $\grid{2}{n}$
		separating $\grid{2}{n}$ into two subgraphs containing at least $n-2$ vertices
		and let~$u_1$ have a neighbor of degree~$3$ in the twisted component.
		Such a separator exists because
		the size of the twisted component is greater than $n-2$.
		Duplicator proceeds with this choice of~$u_1$,~$u_2$, and~$v$ as in the case before\footnote{Formally, our construction of $\bij{v}{e_1}{e_2}$ requires two pebble pairs to be placed.
		If they are not placed already,
		we just pretend that they are placed on~$u_1$ and~$u_2$.}.
		After the current round, the twisted component has size at least $n-4$.
		\qedhere
	\end{itemize}
\end{proof}

With Lemmas~\ref{lem:refinement-implies-logic-step}
and~\ref{lem:walklogic-distinguish-edges-implies-pebble-game}
we obtain the following corollaries:

\begin{corollary}
	\label{cor:lower-bound-walk-logic}
	A walk counting logic formula distinguishing
	$\repl{\grid{2}{n}}$ and $\replt{\grid{2}{n}}$
	has quantifier depth $\Omega(n)$.
\end{corollary}
\begin{corollary}
	\label{cor:lower-bound-walk-refinement}
	Walk refinement distinguishes
	$\repl{\grid{2}{n}}$ and $\replt{\grid{2}{n}}$
	in $\Omega(n)$
	iterations.
	In particular,
	walk refinement stabilizes on
	$\repl{\grid{2}{n}}$ as well as $\replt{\grid{2}{n}}$
	in $\Omega(n)$ iterations.
\end{corollary}

Recall that $\repl{\grid{2}{n}}$ and $\replt{\grid{2}{n}}$
have $\Theta(n)$ vertices, so both corollaries give bounds that are linear in the number of vertices.

We want to remark that $4$ pebble pairs already suffice for Spoiler
to reduce the size of the twisted component by $2$ in each round:
Spoiler places the pebbles always on four vertices originating from a 4-cycle 
starting in the middle of the graph.
Then she moves two of them so that together the 4 pebbles cover a 4-cycle that shares an edge with the 4\nobreakdash -cycle of the previous round. 
Overall, this strategy requires $n/2 + 1$ rounds if $n$ is even and
$(n+1)/2 +1$ rounds otherwise. 
In particular, $4$-walk refinement has the same iteration number as walk-refinement
on these graphs.
Nevertheless, after only one iteration,
$n$-walk refinement distinguishes vertices of different gadgets,
but $4$-walk refinement does not.

With only $3$ pebble pairs the size reduces by at most one per iteration
(as already shown in~\cite{Furer2001}).

\section{Conclusion}

We showed that the $2$-dimensional Weisfeiler-Leman refinement
stabilizes an $n$-vertex graph after $\bigO(n \log n)$ iterations.
Hence in the counting logic $\countinglogic{3}$
we only require a quantifier depth of $\bigO(n \log n)$.
This matches the best known lower bound of the form~$\Omega(n)$ up to a logarithmic factor.
Thus the question remains what the precise bound is,
and whether the iteration number can be superlinear.
At least for the walk refinement we have now matching linear lower and upper bounds.

It remains also an open problem whether our techniques
can be applied to counting first order logic with more than three variables 
(equivalently higher dimensional Weisfeiler-Leman refinement) or to three variable first order logic without counting.

For all of these mentioned avenues of investigation it could be interesting to find a
combinatorial argument for the $\bigO(n)$ bound for walk refinement.
Finally, we also introduced walk counting logic
and the bijective walk pebble game and studying these remains as future work.

\bibliographystyle{plain}
\bibliography{wl2_algebra}

\end{document}